\pdfoutput=1
\documentclass{lmcs}
\pdfoutput=1

\usepackage{lastpage}
\lmcsdoi{20}{4}{25}
\lmcsheading{}{\pageref{LastPage}}{}{}%
{Mar.~29,~2022}{Dec.~12,~2024}{}

\usepackage[utf8]{inputenc}


\usepackage{breakcites}
\usepackage{preamble}
\usepackage{macros}
\usepackage{categories}
\usepackage{jon-tikz}
\usepackage{pftools}
\usepackage[disable,colorinlistoftodos,prependcaption,textsize=scriptsize]{todonotes}
\usepackage{hyperref}

\def\eg{{\em e.g.}}

\begin{document}

\title[Unifying cubical and multimodal type theory]{Unifying cubical and multimodal type theory}

\author[F.~Aagaard]{Frederik Lerbjerg Aagaard%
\lmcsorcid{0000-0002-9132-0098}}[a]

\author[M.~Kristensen]{Magnus Baunsgaard Kristensen}[b]

\author[D.~Gratzer]{Daniel Gratzer%
\lmcsorcid{0000-0003-1944-0789}}[a]

\author[L.~Birkedal]{Lars Birkedal%
\lmcsorcid{0000-0003-1320-0098}}[a]

\address{Aarhus University}
\email{aagaard@cs.au.dk, gratzer@cs.au.dk, birkedal@cs.au.dk}

\address{IT University of Copenhagen}
\email{magnusbk95@gmail.com}

\begin{abstract}
  In this paper we combine the principled approach to modalities from multimodal type
  theory (MTT) with the computationally well-behaved realization of identity types from cubical type
  theory (CTT). The result---cubical modal type theory (Cubical MTT)---has the desirable features of
  both systems. In fact, the whole is more than the sum of its parts: Cubical MTT validates
  desirable extensionality principles for modalities that MTT only supported through ad hoc means.

  We investigate the semantics of Cubical MTT and provide an axiomatic approach to producing
  models of Cubical MTT based on the internal language of topoi and use it to construct presheaf
  models. Finally, we demonstrate the practicality and utility of this axiomatic approach to models
  by constructing a model of (cubical) guarded recursion in a cubical version of the topos of
  trees. We then use this model to justify an axiomatization of L{\"o}b induction and thereby use
  Cubical MTT to smoothly reason about guarded recursion.
\end{abstract}

\maketitle

\section{Introduction}
\label{sec:intro}

Type theorists have produced a plethora of variants of type theory since the introduction of
Martin-L{\"o}f type theory (\MLTT{}), each of which attempts to refine \MLTT{} to enhance its
expressivity or convenience. Unfortunately, even extensions of type theory which appear orthogonal
cannot be carelessly combined. Expert attention is frequently necessary to ensure that combinations
of extensions retain desirable properties of base type theory. We are particularly interested in two
families of extensions to \MLTT{}: \emph{cubical type theories}~\cite{cohen:2018,angiuli:2021} and
\emph{(Fitch-style) modal type theories}~\cite{clouston:2018,birkedal:2020,gratzer:2020}.

Both of these lines of research aim to increase the expressivity of type theory, but along different
axes. Cubical type theory gives a higher-dimensional interpretation to the identity type and thereby
obtains a more flexible notion of equality along with a computational interpretation of
\emph{univalent foundations}~\cite{hottbook}. Meanwhile, modal dependent type theory (\MTT{}) extends
\MLTT{} with connectives which need not commute with substitution, allowing for type theory to model
phenomena such as guarded recursion, axiomatic cohesion, or
parametricity~\cite{birkedal:2012,shulman:2018,nuyts:2018}.

While combining two complex type theories like this is not a task to be undertaken frivolously,
experience has shown that a principled mixture of these two type theories would be useful. Indeed,
modalities naturally appear in synthetic homotopy theory~\cite{shulman:2018} but without an
apparatus like \MTT{}, these extensions must be handled in an ad-hoc way, which can easily disrupt
the desirable properties of type theory. Moreover, cubical type theory's version of equality
validates invaluable principles like function extensionality and a cubical variant of \MTT{} would
thus eliminate the need to postulate such principles~\cite{birkedal:2019,kristensen:2021}.

Prior to discussing how \CubicalMTT{} fuses these systems, we set the stage by introducing both
cubical type theory and multimodal type theory separately.

\subsection*{Cubical type theory}

Cubical type theory originates from the broader class of \emph{homotopy type theories} whose study
was instigated by Voevodsky's observation that the intensional identity type could be realized as a
(homotopical) path space~\cite{kapulkin:2021}. This shift in perspective justifies the inclusion of
the univalence axiom which postulates an equivalence between equalities of elements of a
universe and equivalences of the denoted types. While univalence has many pleasant consequences
(function extensionality, effectivity of quotients, \etc{}), the addition of such an axiom disrupts
crucial properties of the type theory. In particular, it is not possible to compute in such a
theory. In order to rectify this issue, cubical type theory was introduced and shown to
simultaneously support computation and validate the univalence axiom.

Cubical type theory extends \MLTT{} with an interval object $\I$ along with a function space---a
path space---to hypothesize over it. Intuitively, $\I$ abstracts the interval $[0,1]$ and this
connection is enhanced by the addition of operations, \eg{} $0, 1 : \I$. Accordingly, $\I \to A$
classifies lines in $A$ and by iterating we obtain squares, cubes, and arbitrary $n$-cubes in $A$.

While homotopy type theory recasts the identity type from \MLTT{} as a path in a space, cubical type
theory begins with paths and forces them to behave like an identity type. We therefore isolate a
subtype $\Path{A}{a_0}{a_1}$ of paths $\I \to A$ whose values at $0$ and $1$ are $a_0$ and $a_1$. By
further equipping $A$ with \emph{Kan operations}, $\Path{A}{a_0}{a_1}$ becomes a new model for the
identity type. The Kan operations are subtle and complex but without them $\Path{A}{-}{-}$ is not
even an equivalence relation. The flexibility afforded by these Kan operations, however, allows
cubical type theory to support a computational interpretation of univalence.

\begin{rem}
  The path type of a type with Kan operations is not a perfect enhancement of Martin-L{\"o}f's
  identity type: the latter satisfies definitional equalities not enjoyed by the former. However,
  these two types share the same universal property and so we shall use the term identity type to
  refer to both and use \emph{intensional identity type} or \emph{path type} to distinguish them.
\end{rem}

\subsection*{\MTT{} and Fitch-style modal type theories}

We now turn from cubical type theory to modal type theory. Like cubical type theory, the motivations
for modalities---type constructors which do not necessarily respect substitution---are semantic;
many models of type theory have further structure which does not directly commute with
substitution but would still be useful to internalize. For instance, the global sections comonad
of a presheaf model is frequently essential for working internally to the
model~\cite{clouston:2015,licata:2018,shulman:2018}, but it almost never commutes with substitution.

Unfortunately, much of the convenience of \MLTT{} hinges on the fact that all operators do commute
with substitution, so introducing a modality tends to disrupt nearly every property of
importance. In order to cope with this contradiction, modal type theories like
\MTT{}~\cite{gratzer:2020} have carefully isolated classes of modalities which can be safely
incorporated into type theory while preserving properties such as canonicity and
normalization~\cite{gratzer:normalization:2022}. In fact, \MTT{} can be instantiated with an
arbitrary collection of (weak) dependent right adjoints~\cite{birkedal:2020}. The metatheory of
\MTT{} applies irrespective of the choice of mode theory, and therefore \MTT{} can be seen as a
\emph{general} modal type theory, suitable for instantiation with a wide variety of different
modalities to specialize the type theory to capture particular models.

In practice, \MTT{} is parameterized by a mode theory, a 2-category used to describe the modalities
in play. The objects of the mode theory correspond to individual copies of \MLTT{} linked together
by the modalities, the morphisms of the mode theory. The 2-cells of the mode theory induce natural
transformations between modalities. For instance, in order to model the global sections comonad we
pick a mode theory with one mode $m$, one modality $\mu$, and a collection of 2-cells shaping this
modality into a comonad, \eg{}, 2-cells $\Mor{\mu}{\mu \circ \mu}$ and $\Mor{\mu}{\ArrId{m}}$
subject
to several equations. Upon instantiating \MTT{} with this mode theory we obtain a type theory with a
comonad already known to satisfy many important metatheorems.

\subsection*{Towards Cubical \MTT{}}

In \cite{gratzer:2020}, each mode of \MTT{} contains a copy of \MLTT{}. Unfortunately, the type
theory therefore inherits the well-known limitations of \MLTT{}: the intensional identity type is difficult to
work with, function extensionality is not satisfied. One can resolve these issues by adding equality
reflection to \MTT{}, but this disrupts the decidability of type-checking. Moreover, several
modalities arise in the context of homotopy type theory~\cite{shulman:2018} and adapting \MTT{} to
these models requires simply postulating univalence, thereby conferring the same set of
difficulties.

We introduce \CubicalMTT{}, a unification of Cubical type theory and \MTT{}. To a first approximation,
\CubicalMTT{} replays the theory of \MTT{}, after replacing \MLTT{} with cubical type theory. One
thereby obtains a modal type theory with different modes---now copies of cubical type
theory---connected by arbitrary dependent right adjoints. Moreover,
each mode now satisfies univalence and function extensionality.

Beyond this, a computation rule for Kan operations in modal types is needed for computation (and thus for normalization), but it is not immediately well-typed.
Indeed, a key challenge in combining \MTT{} with \CTT{} is exactly to capture sufficient interactions between modal and cubical aspects for this rule to be well-typed, whilst not making greater demands than the intended models can bear.

The switch from using \MLTT{} to using \CTT{} in \MTT{} also improves modal types.  For instance, in
\cite{gratzer:normalization:2022} special care is taken to include \emph{crisp induction} in order
to validate the modal counterpart to function extensionality.  While this addition preserves
normalization and canonicity, modal extensionality is independent of \MTT{}.  In \CubicalMTT{}, by
contrast, the corresponding principle is simply provable (Theorem~\ref{thm:cmtt:modal-ext}).

We show that models of \CubicalMTT{} can be assembled from certain models of cubical type theory
connected by right adjoints. In particular, given coherent functors $\Mor[f_\mu]{\CC_n}{\CC_m}$
there is a model of \CubicalMTT{} which realizes context categories as $\PSH[\CSET]{\CC_m}$ and
modalities as right Kan extension $\RKan*{f_\mu}$. This ensures, for example, that despite the
complexity of both \MTT{} and cubical type theory, it is easily shown that \CubicalMTT{},
appropriately instantiated, models cubical guarded recursion~\cite{birkedal:2019,kristensen:2021}.
Indeed, we show that the cubical underpinnings of \CubicalMTT{} improve the presentation of guarded
recursion in \MTT{}~\cite{gratzer:journal:2021}.

The development of this theory of models also implies the soundness of \CubicalMTT{}. We further
conjecture, but do not prove, that the normalization results of \cite{gratzer:normalization:2022}
for \MTT{} and \cite{sterling:2021} for cubical type theory can be appropriately combined into a
normalization proof for \CubicalMTT{}.

\subsection*{Contributions}

We contribute \CubicalMTT{}, a synthesis of cubical type theory and \MTT{}, and a foundation for
multimodal and higher-dimensional type theories. In \autoref{sec:cubical-and-mtt} we recapitulate
the
basics of cubical type theory and \MTT{} and in \autoref{sec:cubical-mtt}
we present the definition of \CubicalMTT{} and further prove the aforementioned modal extensionality
principle. Finally, in \autoref{sec:semantics} we introduce the model theory of \CubicalMTT{} and
further show that cubical presheaves and certain essential geometric morphisms assemble into models.
We then apply this theory to cubical guarded recursion in \autoref{sec:examples} and explore the
improved presentation of guarded recursion.

\section{Cubical and multimodal type theory}
\label{sec:cubical-and-mtt}

We now recall the essential details of cubical type theory~\cite{cohen:2018} and
\MTT{}~\cite{gratzer:journal:2021}. We focus mostly on the portions relevant for \CubicalMTT{} and
refer the reader to the existing literature for a more thorough introduction. Readers familiar with
both systems may safely proceed to \autoref{sec:cubical-mtt}.

\subsection{Cubical type theory}
\label{sec:cubical-syn}
\CTT{} begins by extending \MLTT{} with a primitive interval $\I$ and algebraic structure to mimic
the real interval $[0,1]$. Terms of type $A$ which assume \emph{dimension variables}
$i, j, k : \I$ correspond to $n$-cubes (lines, squares, cubes) in $A$. Concretely, we add a new
context formation rule $\Gamma, i : \I$ and a new syntactic class of \emph{dimension terms}
$\Gamma \vdash r : \I$:
\begin{grammar}
  Abstract interval & r, s : \I &
  i \GrmSep 0 \GrmSep{} 1 \GrmSep{} 1 - r \GrmSep{} r \lor s \GrmSep{} r \land s
\end{grammar}
We further identify dimension terms by the equations of De Morgan algebras.

\begin{rem}
  We note that there are in fact many variations on cubical type theory which primarily impose
  differing amounts of structure on the interval. For our work, we have opted De Morgan cubical type
  theory~\cite{cohen:2018}, but we expect our work could be adapted to \eg{}, cartesian cubical type
  theory~\cite{angiuli:2021}.
\end{rem}

Next, we add \emph{path types}: a dependent product over the interval. The rules for this new
connective are given in \autoref{fig:ctt-rules} and---just as with dependent products---path types
enjoy $\beta$ and $\eta$ rules stating \eg{} $\PathApp{\prn{\PathAbsV{p}}}{r} = \Sb{p}{r/i}$. In
addition to $\beta$ and $\eta$, paths are equipped with further equalities reducing them at
endpoints, \eg{}, given $p : \Path{A}{a}{b}$ then $\PathApp{p}{0} = a : A$.

\begin{rem}
  Given $x, y : A$, we write $x \PathEq y$ when there exists an element of $\Path{A}{x}{y}$.
\end{rem}

Out of the box, paths define a relation on types which is reflexive and symmetric and which
validates extensionality principles such as function extensionality. They are not, however,
transitive and it is this deficiency that leads to the \emph{Kan composition operation} which forms
the
backbone of \CTT{}. Intuitively, this composition operation lets us complete an open
box (an $n$-cube missing certain faces) to an $n$-cube. In order to formalize this geometric
intuition we add the face lattice $\F$, a class of propositions, which we use to codify the open
boxes to be filled. We therefore add another syntactic class $\Gamma \vdash \phi : \F$:
\begin{grammar}
  Face lattice & \phi,\psi: \F & \bot \GrmSep{} \top \GrmSep{} (r=0) \GrmSep{} (r=1) \GrmSep{} \phi
  \lor \psi \GrmSep{} \phi \land \psi
\end{grammar}
Elements of $\F$ are identified by the equations of distributive lattices as well as the additional
equation $(r=0) \land (r=1) = \bot$.

Elements $\phi : \F$ are used to \emph{restrict} a context by assuming
them, denoted $\RCxV{\Gamma}$. This allows us to locally force $\phi$ to hold so that, \eg{},
$\RCxV {i:\I} [(i=0)] \vdash i = 0 : \I$ \footnote{Note that $(i = 0)$ on the left-hand side is a
face restriction, whilst $i = 0$ on the right-hand side is a definitional equality.}.
We can take advantage of an assumption $\phi$ in our context through \emph{systems}. The rules for
systems are given in \autoref{fig:ctt-rules}; intuitively they state that to construct an element in
$\RCxV{\Gamma}[\bigvee_i \phi_i] \vdash u : A$, it suffices to construct elements
$\RCxV{\Gamma}[\phi_i] \vdash u_i : A$ that agree on the overlap. An element constructed through
this amalgamation restricts appropriately \eg{}, $\RCxV{\Gamma}[\phi_i] \vdash u = u_i : A$.

We are frequently concerned with the behavior of a term after some assumption $\phi$---its
\emph{boundary}---and therefore introduce notation for it. We write
$\Gamma \vdash a : \ExtType {A} {\phi\mapsto u}$ as shorthand for (1) $a$ being a term of $A$ and
(2) under the assumption $\phi$, $a = u : A$. With this machinery, we can now formulate the Kan
composition rule, shown in \autoref{fig:ctt-rules}. This one principle is sufficient to prove the
properties we expect of identity types, including J (path induction).\footnote{Unlike in \MLTT{},
  however, path induction reduces on reflexivity only up to a path.}

\begin{figure}
  \begin{mathparpagebreakable}
    \inferrule{
      \Gamma \vdash a : A
      \\
      \Gamma \vdash b : A
    }{
      \Gamma \vdash \Path A a b
    }
    \and
    \inferrule{
      \Gamma, i : \I \vdash p : A
    }{
      \Gamma \vdash \PathAbsV{p} : \Path A {p(0)} {p(1)}
    }
    \and
    \inferrule{
      \Gamma \vdash r : \I
      \\
      \Gamma \vdash p : \Path A a b
    }{
      \Gamma \vdash \PathApp p r : A
    }
    \and
    \inferrule{
      \Gamma, \phi_i \vdash a_i : A
      \\
      \Gamma, \phi_0 \land \phi_1 \vdash a_0 = a_1 : A
    }{
      \RCxV{\Gamma}[\phi_0 \lor \phi_1] \vdash \SysBin {\phi_0} {a_0} {\phi_1} {a_1} : A
    }
    \and
    \inferrule{ }{
      \RCxV{\Gamma}[\Fbot] \vdash \SysNull : A
    }
    \and
    \inferrule{
      \Gamma \vdash \phi : \F
      \\
      \Gamma, i : \I \vdash A
      \\
      \Gamma, \phi, i : \I \vdash u : A
      \\
      \Gamma \vdash u_0 : \ExtType{\Sb {A}{0/i}}{\phi \mapsto \Sb{u}{0/i}}
    }{
      \Gamma \vdash \CompCommV{i}{A}{\phi \mapsto u}{u_0} : \ExtType{\Sb{A}{1/i}}{\phi\mapsto \Sb
      {u}{1/i}}
    }
  \end{mathparpagebreakable}
  \caption{Selected rules from \CTT{}}
  \label{fig:ctt-rules}
\end{figure}

We review the proof that path equality is transitive to give the reader a sense of the rule.
Let $A$ be a type that does not depend on any interval variables, and suppose $a,b,c : A$,
$p:\Path{A}{a}{b}$, and $q:\Path{A}{b}{c}$. We form three lines in $A$: The paths $p$ and $q$ as
well as the constant $a$ line. Using these we form a system depending on $i$ and $j$ given by
$\SysBin {(i=0)} {a} {(i=1)} {\PathApp {q} {j}}$. The path $p$ forms an extension of this system at
$j=0$, and so we can form the path
$\PathAbsV {\CompCommV {j} {A} {(i=0) \mapsto a, (i=1) \mapsto \PathApp {q} {j}} {\PathApp {p}
    {i}}}$, which will reduce to the $j=1$ part of our designed system, i.e., $a$ at $i=0$ and
$\PathApp {q} {1} = c$ at $i=1$, thus proving transitivity. We think of the input data as an open
box with bottom $p$ and sides given by the system; in this analogy the composition forms a lid
completing the outer square.
\[\begin{tikzpicture}[diagram]
  \node (NW) {$a$};
  \node [below = of NW] (SW) {$a$};
  \node [right = of SW] (SE) {$b$};
  \node [above = of SE] (NE) {$c$};
  \path [->] (SW) edge node [left] {$a$} (NW);
  \path [->] (SW) edge node [below] {$\PathApp{p}{i}$} (SE);
  \path [->] (SE) edge node [right] {$\PathApp{q}{j}$} (NE);
  \path [dashed,->] (NW) edge (NE);
\end{tikzpicture}\]

\subsection{Multimodal type theory}
\label{sec:multimodal-syn}

\begin{figure}
  \begin{mathparpagebreakable}
    \inferrule{
      \IsCx \Gamma <m>
    }{
      \IsCx {\LockCxV \Gamma} <n>
    }
    \and
    \inferrule{
      \IsTy [\LockCxV \Gamma] A <n>
    }{
      \IsTy {\Modify{A}}
    }
    \and
    \inferrule{
      \IsTm [\LockCxV \Gamma] a A <n>
    }{
      \IsTm {\MkBox a} {\Modify{A}}
    }
    \and
    \inferrule{
      \IsTy [\LockCxV \Gamma] A <n>
    }{
      \IsCx {\ECxV \Gamma x A}
      \\
      \IsTm [\LockCxV {\ECxV \Gamma x A}] {x} {A} <n>
    }
    \and
    \inferrule{
      \mu : n \to m \and \nu : o \to n \and \IsCx \Gamma
      \\
      \IsTy [\LockCxV {\LockCxV \Gamma} <\nu>] A <o>
      \\
      \IsTm [\LockCxV \Gamma] a {\Modify [\nu] {A}} <n>
      \\
      \IsTy [\ECxV \Gamma x {\Modify [\nu] {A}}] B
      \\
      \IsTm [\ECxV \Gamma y A <\mu \circ \nu>] b {\Sb B {\MkBox [\nu] y / x}}
    }{
      \IsTm {\LetMod{a}[y]{b}<\nu>[\mu]} {\Sb B {a / x}}
    }
  \end{mathparpagebreakable}
  \caption{Selected rules from \MTT{}}
  \label{fig:mtt-rules}
\end{figure}

To a first approximation, \MTT{} is a collection of copies of \MLTT{} for each $m : \Mode$,
connected by dependent adjunctions~\cite{birkedal:2020}. \MTT{} is parameterized by a mode
theory $\Mode$~\cite{licata:2016}, a strict 2-category. Each object $m,n,o : \Mode$ is assigned
to a distinct \emph{mode}: a copy of \MLTT{} complete with its own judgments ($\IsCx{\Gamma}$,
$\IsTm{M}{A}$, \dots). Many of the rules of \MTT{} (dependent sums, inductive types, \etc{}) are
\emph{mode-local} and taken as-is from \MLTT{}; the interesting features of \MTT{}
arise from allowing modes to interact with each other.

Modes are connected to each other by \emph{modalities}: a morphism $\Mor[\mu]{n}{m}$
induces a modality sending types $A$ from mode $n$ to types $\Modify{A}$ in mode $m$. The actual
presentation of modalities is necessarily complex because of dependence: given a type $\IsTy{A}<n>$,
there is no obvious choice of context in mode $m$ for $\Modify{A}$. \MTT{} resolves this tension in
Fitch-style~\cite{clouston:2018} and pairs each modality with an adjoint action on contexts. In
particular, given a modality $\Mor[\mu]{n}{m}$, we can obtain a new context
$\IsCx{\LockCxV{\Gamma}}<n>$
from $\IsCx{\Gamma}$. Further rules turn $\LockCxV{-}$ into a functor between
categories of contexts and substitutions at modes $n$ and $m$; intuitively a left
adjoint to the modal type former $\Modify{-}$. See the introduction and formation rules for
$\Modify{-}$ recorded in \autoref{fig:mtt-rules}.

The elimination rule for $\Modify{-}$ is complex because we cannot `reverse' the introduction rules
without violating the admissibility of substitution. Instead, \MTT{} annotates each variable
in the context and replaces $\Gamma, x : A$ with $\ECxV{\Gamma}{x}{A}$. One can access a variable
annotated with $\mu$ if and only if it appears behind precisely $\LockCxV{-}<\mu>$. The elimination
rule uses these annotations to implement a \emph{modal induction principle} and allows one to
reduce the
process of constructing an element of $\Sb{B}{a/x}$ for some
$\IsTm[\LockCxV{\Gamma}<\nu>]{a}{\Modify{A}}$ to the case $\Sb{B}{\MkBox{y}/x}$ for some fresh
$\DeclVar{y}{A}<\nu\circ\mu>$; see the precise rule in \autoref{fig:mtt-rules}.

Thus far we have not mentioned the (2-)categorical structure of $\Mode$, but it is this additional
structure which allows us to control the behavior of modalities. In fact, the operation sending a
modality $\mu$ to $\LockCxV{-}$ is 2-functorial so that, \eg{},
$\LockCxV{\LockCxV{\Gamma}}<\nu> = \LockCxV{\Gamma}<\mu\circ\nu>$. This fact is reflected into
types; the assignment $\mu \mapsto \Modify[\mu]{-}$ is essentially pseudo-functorial.
Consequently, a 2-cell
$\Mor[\alpha]{\mu}{\nu}$ in $\Mode$ induces a substitution
$\IsSb[\LockCxV{\Gamma}<\nu>]{\Key{\alpha}{\Gamma}}{\LockCxV{\Gamma}}$ which in turn produces a
collection
of functions $\Modify[\mu]{-} \to \Modify[\nu]{-}$. By modifying the equalities and 2-cells of
$\Mode$, we can force $\Modify{-}$ to become, \eg{}, a comonad, a right adjoint, \etc{}

\MTT{} also extends dependent products to hypothesize over types annotated with modalities
other than $\ArrId{}$, \ie{}, to abstract over
$x : \prn{\DeclNameless{A}}$~\cite{gratzer:journal:2021}. While these \emph{modal dependent
products} are convenient, we refrain from discussing them here for simplicity.

\section{\texorpdfstring{\CubicalMTT}{Cubical MTT}}
\label{sec:cubical-mtt}

Cubical multimodal type theory (\CubicalMTT{}) enhances \MTT{} with a better behaved identity type
and univalence by combining it with \CTT{}. Like \MTT{}, \CubicalMTT{} is parameterised by a mode
theory, \ie{}, a 2-category of modes, modalities, and 2-cells. Whereas \MTT{} has a copy of \MLTT{}
at each mode, \CubicalMTT{} has a copy of \CTT{}. A guiding principle in the design of \CubicalMTT{}
is that cubical and modal aspects should be \emph{orthogonal} to each other. In particular, the
primitives of each system should interact as little as possible with primitives from the other. To
realize this, we add certain exchange principles in~\autoref{sec:cub-exchange} which are then
applied in~\autoref{sec:modal-comp} to define composition for modal types.

We detail the novel aspects of \CubicalMTT{} and refer to \autoref{app:rules} for an exhaustive
account.

\begin{figure}
  \begin{mathparpagebreakable}
    \inferH{int/exc}{
      \IsTm{r}{\I _ m}
    }{
      \IsTm[\LockCxV \Gamma]{\ExcIntInv r}{\I_n}<n>
    }
    \and
    \inferH{face/exc}{
      \IsTm  \phi {\F _ m} <m>
    }{
      \IsTm  [\LockCxV \Gamma] {\ExcCofInv \phi} {\F _ n} <n>
    }
    \\
    \inferH{sb/exc-int}{
    }{
      \IsSb [\LockCxV {\ICxV \Gamma}] \ExcSbI {\ICxV {\LockCxV \Gamma} <n>}
    }
    \and
    \inferH{sb/exc-face}{
      \IsTm \phi {\F _ m}
    }{
      \IsSb [\LockCxV {\RCxV \Gamma}] \ExcSbR {\RCxV {\LockCxV \Gamma} [\ExcCofInv \phi]}
    }
    \and
    \inferH{sb/exc-int-inv}{
    }{
      \IsSb [\ICxV {\LockCxV \Gamma} <n>] \ExcSbIinv {\LockCxV {\ICxV \Gamma}} <n>
    }
    \and
    \inferH{sb/exc-face-inv}{
      \IsTm \phi {\F _ m}
    }{
      \IsSb [\RCxV {\LockCxV \Gamma} [\ExcCofInv \phi]] \ExcSbRinv {\LockCxV {\RCxV \Gamma}} <n>
    }
    \and
    \inferH{term-eq/comp-mod}{
      \IsTy [\LockCxV {\ICxV \Gamma}] A <n>
      \\
      \IsTm \phi {\F _ m}
      \\
      \IsTm [\LockCxV {\ICxV {\RCxV \Gamma}}] u A <n>
      \\
      \IsTm [\LockCxV \Gamma] {u _ 0} {\Sb A {0 / i}} <n>
      \\
      \EqTm [\LockCxV {\RCxV \Gamma}] {\Sb u {0 / i}} {u _ 0} {\Sb A {0 / i}} <n>
    }{
      \Gamma
      \vdash
      {\small \begin{array}{@{}c@{}}
        \CompCommV i {\Modify {A}} {\phi \mapsto \MkBox u} {\MkBox {u _ 0}}
        \\
        =
        \\
        \MkBox{\CompCommV i A {\ExcCofInv \phi \mapsto \Sb u {\ExcSbIinv \circ \ExcSbRinv}} {u _ 0}}
      \end{array}}
      :
      \Sb {\Modify {A}} {1 / i}
      \Mute{{} \mathop{@} m}
    }
  \end{mathparpagebreakable}
  \caption{Selected rules of \CubicalMTT{}, presupposing $\Mor[\mu]{n}{m}$ and $\IsCx \Gamma$.}
  \label{fig:mtt-box-rules}
\end{figure}

\subsection{Cubical exchange}
\label{sec:cub-exchange}

The orthogonality principle of \CubicalMTT{} dictates that the interval should be minimally
impacted by the action of a modality on the context. Accordingly, we add \emph{exchange
operations}. Concretely, given a dimension term $\IsTm{r}{\I_m}$, we add a new dimension term
$\IsTm[\LockCxV{\Gamma}]{\ExcIntInv{r}}{\I_n}<n>$, see \ruleref{int/exc} in
\autoref{fig:mtt-box-rules}. We demand that this operation is a morphism of De Morgan algebras and
is lax natural, e.g. $\ExcIntInv{\prn{r \land s}} = \ExcIntInv{r} \land \ExcIntInv{s}$ and
$\ExcIntInv{r}<\mu\circ\nu> = \ExcIntInv{\prn{\ExcIntInv{r}}}<\nu>$. Using this operation, it is
possible to derive the \emph{exchange substitution}
$\IsSb[\LockCxV{\ICxV{\Gamma}}]{\ExcSbI}{\ICxV{\LockCxV{\Gamma}}<n>}$, see \ruleref{sb/exc-int}.
Finally, we add an inverse to this, see \ruleref{sb/exc-int-inv}, making $\LockCxV{\ICxV{\Gamma}}$
isomorphic to $\ICxV{\LockCxV{\Gamma}}<n>$, once again in accordance with the orthogonality
principle.

The case is entirely symmetrical for elements of the face lattice and the corresponding restricted
contexts. Concretely, given a face $\IsTm{\phi}{\F_m}$, we add a new face
$\IsTm[\LockCxV{\Gamma}]{\ExcCofInv{\phi}}{\F_n}<n>$, see \ruleref{face/exc}. Similarly to before,
we require this operation to be a morphism of bounded lattices and be lax natural, but we further
require that it matches with the corresponding operation on the interval via the equation
$\ExcCofInv{\EqBot{r}} = \EqBot{\ExcIntInv{r}}$. Thus, $\ExcCofInv{-}$ commutes with everything but
dimension variables, meaning that $\ExcCofInv{\phi}$ is precisely $\ExcIntInv{-}$ applied to every
dimension variable in $\phi$. Just as before, we can derive a substitution, to which we add an
inverse, making $\LockCxV{\RCxV{\Gamma}}$ and $\RCxV{\LockCxV{\Gamma}}[\ExcCofInv{\phi}]$
isomorphic, see \ruleref{sb/exc-face} and \ruleref{sb/exc-face-inv}.

As we will see shortly, these rules are sufficient for composition in modal types, but one may
still wonder if there would be merit in the addition of inverses to $\ExcIntInv{r}$ and
$\ExcCofInv{\phi}$; after all, this would be in line with our orthogonality principle. It turns out,
however, that such an addition would lead to a significant restriction of what models are valid, in
particular, it would invalidate our model of guarded recursion in \autoref{sec:examples}, and we thus
refrain from making such an addition.

As mentioned, the exchange operations respect the 2-categorical structure of the mode theory, and
since the exchange substitutions are derived from the simpler exchange operations, they inherit this
property. We now record one such coherence explicitly for future use. \CubicalMTT{} inherits a
\emph{weakening substitution} from \CTT{}: $\Mor[\WkF[\phi]]{\RCxV{\Gamma}}{\Gamma}$. One may show
that the two canonical substitutions
$\Mor{\RCxV{\LockCxV{\Gamma}}[\ExcCofInv{\phi}]}{\LockCxV{\Gamma}}$ are equal. Explicitly, we
equate the direct restriction substitution $\WkF[\ExcIntInv{\phi}]$ to $\LockSb{\WkF[\phi]}\circ
\ExcSbRinv$.

\subsection{Composition in modal types}
\label{sec:modal-comp}

Now we can tackle the problem of composition in \CubicalMTT{}. Composition is, as the other cubical
rules, added to the system locally and satisfies the same computation rules familiar from \CTT{}
for standard type formers. Modal types will support a computation principle similar to that of
inductive types, allowing us to commute $\MkBox{-}$ with $\comp$. Thus the status
of composition in modal types is similar to that of natural numbers, where composition is a formal
operation that reduces on canonical forms, as opposed to \eg{} dependent sums.

The desired `reduction' is \ruleref{term-eq/comp-mod}. We take a moment to show that the conclusion
of this rule is well-typed and that the result has the expected boundary.

Inspecting the assumptions of this rule, we note that all but one are equivalent to a
composition problem in $\Modify{A}$ where the input terms are of form $\MkBox{u}$ and
$\MkBox{u_0}$---with the exception that the assumption $\Sb{u}{0 / i} = u_0$ is slightly stronger
than
necessary---so it is clear that the left-hand side of this equality is well-typed. That the
right-hand side,
$\MkBox{\CompCommV{i}{A}{\ExcCofInv{\phi} \mapsto \Sb{u}{\ExcSbI \circ \ExcSbR}}{u_0}}$, is
well-typed is more subtle. Inspecting the rule for composition in \autoref{fig:ctt-rules}, we see
that we must first verify the following:
\begin{enumerate}
\item $\IsTm[\LockCxV{\Gamma}]{\ExcCofInv{\phi}}{\F_n}<n>$
\item $\IsTm[\ICxV{\RCxV{\LockCxV \Gamma}[\ExcCofInv{\phi}]}]{\Sb{u}{\ExcSbIinv \circ
\ExcSbRinv}}{\Sb{A}{\ExcSbIinv}}<n>$
\item $\IsTm[\LockCxV{\Gamma}]{u_0}{\Sb{A}{0/i}} <n>$
\item\label{item:iv} $\EqTm[\LockCxV{\RCxV \Gamma}]{\Sb{u}{0/i}}{u_0}{\Sb{A}{0 / i}}<n>$
\end{enumerate}
All of these are immediate results of the premises of \ruleref{term-eq/comp-mod}. In particular,
\autoref{item:iv} is precisely the aforementioned stronger premise.

Assured that both sides of \ruleref{term-eq/comp-mod} are well-typed, we show that that the
right-hand side of this equality satisfies the same boundary condition as the left-hand side, \ie{},
that the right-hand side is equal to $\Sb{\MkBox{u}}{1/i}$ under $\phi$.\footnote{%
  This requirement is a further sanity check on the rule; without this equality the right-hand side
  would not solve the same composition problem as the left and the equation would be highly suspect.
}

First, we observe that in context $\IsCx{\RCxV {\LockCxV{\Gamma}}[\ExcCofInv{\phi}]}<n>$ we have the
following:
\[
  \CompCommV{i}{A}{\ExcCofInv{\phi} \mapsto \Sb{u}{\ExcSbIinv \circ \ExcSbRinv}}{u_0}
  =
  \Sb{
    \Sb{u}{\ExcSbIinv \circ \ExcSbRinv}
  }{1 / i}
  =
  \Sb{\Sb{u} {1 / i}}{\ExcSbRinv}
\]

Next, we recall that weakening by an assumption of the face lattice commutes with face
exchange. Given that the former is silent in our notation and the latter is not, this leads to the
somewhat odd equation $\EqTm[\RCxV{\Gamma}]{\MkBox{m}}{\MkBox{\Sb{m}{\ExcSbR}}}{\Modify{A}}$ when
$\IsTm[\LockCxV{\Gamma}]{m}{A}$. Combining these two equations, we have the following in context
$\RCxV{\Gamma}$:
\begin{align*}
  \MkBox{\Sb{u}{1/i}} & =
  \MkBox{
    \Sb{
      \prn{\CompCommV{i}{A}{\ExcCofInv{\phi} \mapsto \Sb{u}{\ExcSbIinv \circ \ExcSbRinv}}{u_0}}
    }{
      \ExcSbR
    }
  } \\
  & = \MkBox{
    \CompCommV{i}{A}{\ExcCofInv \phi \mapsto \Sb u {\ExcSbIinv \circ \ExcSbRinv}}{u_0}
  }
\end{align*}

\begin{rem}
  The composition rule for modal $\Pi$-types is the same as for non-modal $\Pi$-types up to
  application of $\LockCxV{-}$ and use of exchange substitutions to make it well-typed.
\end{rem}

\subsection{Extensionality principles in \texorpdfstring{\CubicalMTT{}}{Cubical MTT}}
\label{sec:modal-ext}

Function extensionality in \CubicalMTT{} follows directly from function extensionality in \CTT{},
since the rules used in \CTT{} are all available mode-locally in \CubicalMTT{}.
We will prove \emph{modal extensionality}, which cannot be proven in \MTT{}:

\begin{thm}
  \label{thm:cmtt:modal-ext}
  Given $\mu : n \to m$ and $\IsTm {a , b} {A} <n>$, where $A$ is
  classified by the universe, there is an equivalence $\ModExt{a}{b} : \Modify{\Path{A}{a}{b}}
  \Equiv \Path{\Modify{A}}{\MkBox{a}}{\MkBox{b}}$.
\end{thm}
\begin{proof}
  We define the map $\ModExt{a}{b}$ and show it to be an equivalence by constructing an inverse.
  Fix $\IsTm{m}{\Modify{\Path{A}{a}{b}}}$ and $\IsTm{r}{\I_m}$. We wish to construct
  $\ModExt{a}{b}\prn{m}\prn{r}$.
  By modal induction it suffices to consider the case where
  $m = \MkBox{p}$ for some $\IsTm [\LockCxV \Gamma] p {\Path {A} a b} <n>$.
  Because $p$ lives in a locked context whereas $r$ does not, we need an exchange operation.
  We form $\IsTm [\LockCxV \Gamma] {\ExcIntInv r} {\I _ n} <n>$, and define
  $\PathApp {\App {\ModExt a b} m} r = \MkBox {\PathApp p {\ExcIntInv r}}$. Towards verifying
  that we obtain an inverse using path induction, note that for $\IsTm [\LockCxV \Gamma] c
  {A} <n>$ we have that
  \[
    \App {\ModExt c c} {\MkBox {\Refl c}} = \Refl {\MkBox c}.
  \]

  Next we define a map $\ModExtInv{a}{b}$ in the inverse direction.\footnote{We will only need path
    induction and modal induction rather than path abstraction to define $\ModExtInv{a}{b}$, meaning
    that it can also be defined in \MTT{}.} By based path induction along with careful modal
  induction, it suffices to define only $\App{\ModExtInv{a}{a}}{\Refl{\MkBox{a}}}$. In this case we
  define $\App {\ModExtInv{a}{a}}{\Refl{\MkBox{a}}} = \MkBox {\Refl a}$.
  For later use, we calculate for $\IsTm [\LockCxV \Gamma] c {A} <n>$ that
  \[
    \App {\ModExtInv{c}{c}} {\Refl {\MkBox c}} \PathEq \MkBox {\Refl c}.
  \]
  We obtain only a path rather than a judgmental equality because path induction computes
  only up to a higher path in cubical type theory.

  Lastly, we prove that these maps form an equivalence. Let
  $\IsTm m {\Modify {\Path {A} a b}}$. We are to find a path between
  $\App {\ModExtInv a b} {\App {\ModExt a b} m}$ and $m$. It suffices to do so when
  $m = \MkBox {\Refl c}$ for some $\IsTm [\LockCxV \Gamma] c {A} <n>$ where we compute:
  \[
    \App {\ModExtInv c c} {\App {\ModExt c c} m}
    = \App {\ModExtInv c c} {\Refl {\MkBox c}} \PathEq m.
  \]

  For the reverse direction, let $\IsTm p {\Path {\Modify {A}} {\MkBox a} {\MkBox b}}$. We need a
  path between $p$ and $\App {\ModExt a b} {\App {\ModExtInv a b} p}$. We again reduce to the case
  where $p = \Refl {\MkBox c}$ and compute from there:
  $\App {\ModExt c c} {\App {\ModExtInv c c} p} \PathEq \App {\ModExt c c} {\MkBox {\Refl c}} = p$.
\end{proof}

\section{Semantics of \texorpdfstring{\CubicalMTT}{Cubical MTT}}
\label{sec:semantics}

Section~\ref{sec:cubical-mtt} toured through \CubicalMTT{} informally, but in fact, \CubicalMTT{}
can be
presented as a particular \emph{generalized algebraic theory}~\cite{cartmell:1978}. This
automatically gives rise to a category of models---a variant of the standard \emph{categories with
  families}~\cite{dybjer:1996}---with several desirable properties such as the initiality of
syntax. However, \CubicalMTT{} is complex and the induced definition of model is nearly intractable
to manipulate, let alone construct.

We fracture the definition of model into more manageable pieces, making heavy use of the natural
models of \MTT{}~\cite{awodey:2018,gratzer:journal:2021}. In order to construct these models, we
introduce \emph{cubical \MTT{} cosmoi}. This is a more compact structure supplementing \MTT{}
cosmoi~\cite{gratzer:normalization:2022} with the ingredients necessary to internally construct a model of
\CTT{}~\cite{orton:2018,licata:2018}. In practice, cosmoi are easier to obtain and suffice for the
most important models \eg{} those in cubical presheaves.

\subsection{Models of \texorpdfstring{\CubicalMTT}{Cubical MTT}}
\label{sec:semantics:models}

We now present the definition of a model of \CubicalMTT{} with mode theory $\Mode$. To begin with,
we require a strict 2-functor $\Mor[\Interp{-}]{\Coop{\Mode}}{\CAT}$, known as the \emph{modal
context structure}. Intuitively, this 2-functor assigns each mode to a category of contexts. From
this viewpoint, the functor $\Mor[\Interp{\mu}]{\Interp{m}}{\Interp{n}}$ sends a morphism
$\Mor[\mu]{n}{m}$ to the adjoint action $\LockCx{-}$ contexts and the 2-cells $\Interp{\alpha}$
interpret the natural transformations $\Key{\alpha}{}$. We now specify the remaining structure on
top of this functor.

\paragraph*{Mode-local structure}

Each mode $\Interp{m}$ should contain a complete model of \CTT{}, and we specify this in the
language of natural models~\cite{awodey:2018} which provides a concise description of the
connectives of type theory.

As a model of \CTT{}, $\Interp{m}$ has an interval object $\IntOb{m} : \Interp{m}$. Just as in
\CTT{}, we require that $\IntOb{m}$ is a De Morgan algebra and that all products $- \times
\IntOb{m}$ exist.

Next, we require a pair of presheaves $\SynTm{m},\SynTy{m} : \PSH{\Interp{m}}$ representing
respectively the collection of terms and types in a given context. Moreover, there is a projection
map $\Mor[\SynEl{m}]{\SynTm{m}}{\SynTy{m}}$ which sends a term to its type. This universe is closed
under dependent sums, products, \etc{} Each mode also contains an interpretation of the face lattice
$\FPSh : \PSH{\Interp{m}}$ and face restriction which is used to specify the composition
operations. While complex, this piece of the model is unchanged from \CTT{} so we relegate further
details to \autoref{app:models}.

\paragraph*{Modal types}
Next we turn to the modal aspect of a model: modal context extension and modal types. Both of these
structures are specified exactly as in \MTT{}~\cite{gratzer:journal:2021}, with the small caveat
that we require an additional equality for composition operation on modal types.

\paragraph*{Cubical exchange}
Finally, we must address the interaction of the functors $\Interp{\mu}$ and the intervals and face
lattices. Mirroring the syntax, we require natural transformations
$\Mor[\IntArr{\mu}]{\Yo{\IntOb{m}}}{\Pre{\Interp{\mu}}{\Yo{\IntOb{n}}}}$ that are pointwise
morphisms of De Morgan algebras and that assemble with $\Yo{\IntOb{m}}$ into a lax natural
transformation. From this, we can define a morphism, which we require to be have an inverse:
\[
  \Mor[\Pair{\Interp{\mu}\pi_1}{\IntArr{\mu,\Gamma\times\IntOb{m}}\prn{\pi_2}}]
  {\Interp{\mu}\prn{\Gamma\times\IntOb{m}}}{\Interp{\mu}\prn{\Gamma}\times\IntOb{n}}
\]

The above is replayed for face lattices: We require natural transformations
$\Mor[\FPSH*{\mu}]{\FPSH{m}}{\Pre{\Interp{\mu}}{\FPSH{n}}}$ that are pointwise morphisms of bounded
lattices and that assemble with $\FPSH{m}$ into a lax natural transformation. From this can be
defined a canonical morphism
$\Mor{\Interp{\mu}\prn{\Restr{m}{\Gamma}{\phi}}}{\Restr{n}{\Interp{\mu}}{\FPSH*{\mu}{\Gamma}{\phi}}}$,
 which we require has an inverse.

\subsection{Cubical \MTT{} cosmoi}
\label{sec:semantics:cosmoi}

Even after the repackaging of models detailed in \autoref{sec:semantics:models}, a model of
\CubicalMTT{} is still a complex object. There are two orthogonal aspects to this complexity: (1)
constructing the models of cubical type theory in each mode and (2) constructing the network of
modalities and their actions on contexts. Fortunately, there already exists a technique to simplify
(1); rather than construct a model of cubical type theory directly, \cite{orton:2018} and
\cite{licata:2018} have shown that any topos satisfying a handful of axioms supports a model of
cubical type theory. Moreover, (2) is partially addressed in \cite{gratzer:normalization:2022} by the notion of an
\MTT{} cosmos which abstracts several of the difficulties of constructing a model of \MTT{}. We now
unify these two ideas to define \emph{cubical \MTT{} cosmoi} and prove that they induce a model of
\CubicalMTT{}.

\paragraph*{\MTT{} cosmoi}

We will first recall the definition of \MTT{} cosmoi and prove that they induce models of \MTT{}.

\begin{defi}
  A cosmos is a pseudofunctor $\Mor[F]{\Mode}{\CAT}$ that takes objects to locally cartesian closed
  categories and morphisms to right adjoints. We denote the left adjoint to $F \prn \mu$ by
  $F _ !  \prn \mu$.
\end{defi}

A cosmos abstracts from the basic situation we encountered in \autoref{sec:semantics:models}: a
2-functor $F$ picking out categories of contexts and the actions of modalities between them. In this
case, we were primarily concerned not with the category $F\prn{m}$, but with \emph{presheaves over
  $F\prn{m}$}. After all, it is the category of presheaves which hosts types and terms and where we
formulate structures like context extension. Careful inspection reveals that we only require the
locally Cartesian closed structure of $\PSH{F\prn{m}}$ when formulating the rest of the structure
of a model, so it is natural to require only that each mode of a cosmos is locally Cartesian
closed. Indeed, on top of this skeleton we can transport more of the structure of a model to cosmoi:

\begin{defi}
  An extensional \MTT{} cosmos is a cosmos $F$ such that each mode is equipped with a morphism
  $\Mor[\SynEl m]{\SynTm m}{\SynTy m}$ inducing a universe closed under dependent products, sums,
  booleans, and extensional identity types. We further require that each map
  $\Mor[F\prn{\mu}]{F\prn{n}}{F\prn{m}}$ induce a \emph{dependent right
    adjoint}~\cite{birkedal:2020,gratzer:journal:2021}.\footnote{These are essentially the same
    as modal types in \MTT{}, further equipped with a syntactically ill-behaved but semantically
    convenient elimination rule.}
\end{defi}

We have leveraged the same intuition as natural models to regard $\SynTy m$
(respectively $\SynTm{m}$)
as the collection of types (resp. terms), but without any representability requirements (they cannot
be stated in LCCCs). Requiring closure of these universes under the connectives of \MTT{} then
ensures that an \MTT{} cosmos induces a model of \MTT{} in the sense of
\cite{gratzer:journal:2021}. Prior to proving this, however, we require the following standard
category-theoretic fact:
\begin{lem}\label{lemma:adjoint_pseudofunctor}
  Let $\CC$ be a 2-category and $\Mor[F]{\CC}{\CAT}$ be a pseudofunctor such that each $F\prn{f}$ is
  a right adjoint $F_!\prn{f} \Adjoint F\prn{f}$ then the left adjoints extend to a pseudofunctor
  $\Mor[F_!]{\Coop{\CC}}{\CAT}$.
\end{lem}

\begin{thm}\label{thm:MTT_cosmos}
  An extensional \MTT{} cosmos induces a model of extensional \MTT{} with modal context structure,
  which is pseudonaturally equivalent to the pseudofunctor of left adjoints induced by
  Lemma~\ref{lemma:adjoint_pseudofunctor}. If the pseudofunctor of left adjoints is a strict
  2-functor, the modal context structure may instead be chosen to be equal to it.
\end{thm}
\begin{proof}
  Fix an extensional \MTT{} cosmos $\Mor[F]{\Mode}{\CAT}$.  By
  Lemma~\ref{lemma:adjoint_pseudofunctor},
  the left adjoints $F _ ! \prn \mu$ assemble into a pseudofunctor $\Mor[F _ !]{\Coop
    \Mode}{\CAT}$. We may strictify this functor to get a strict 2-functor
  $\Mor[\widehat F _ !]{\Coop \Mode}{\CAT}$ and a pseudonatural equivalence of categories
  $\Mor[\alpha]{F _ !}{\widehat F _ !}$. We claim that $\widehat{F}_!$ models extensional \MTT{}.

  For each mode $m$, we define the universe of types and terms as the Yoneda embedding of the
  universe already present in $\widehat{F}_!\prn{m}$:
  $\widehat {\SynEl m} = \Yo {\alpha _ m \prn {\SynEl m}}$.  Because $\widehat{F}_!\prn{m}$ is
  finitely complete, this is a representable natural transformation. Moreover, since both $\alpha_m$
  and $\Yo$ preserve LCCC structure, this universe is closed under the types in the types cosmos:
  dependent right adjoints for each $\mu$, dependent sums, products, booleans, and extensional
  identity types \etc{} Thus by \cite[Theorem
  7.1]{gratzer:journal:2021} $\widehat{F}_!$
  models
  extensional \MTT{}.
\end{proof}

\paragraph*{Cubical cosmoi}

An \MTT{} cosmos $F$ interprets each mode as an LCCC $F\prn{m}$ because locally Cartesian closed
structure is sufficient to specify the connectives of \MTT{}. Unfortunately, it is not sufficient to
apply the techniques of \cite{orton:2018} and \cite{licata:2018} and internally construct a model of
\CTT{}. We therefore isolate the notion of \emph{a LOPS topos},\footnote{Named after the authors of
  \cite{licata:2018}} containing precisely the required structure. We further define cubical cosmoi
as particularly well-behaved networks of LOPS topoi.

\begin{defi}
  A LOPS topos is an elementary topos $\EE$ with a hierarchy of universes, an object of cofibrations
  $\EmbMor{\LOPSFace _ {\EE}}{\Omega}$ and a tiny interval object $\LOPSInt _ {\EE}$ subject to the
  Orton-Pitts axioms.\footnote{%
    In fact, we make use of a slight strengthening of axioms presented by \cite{orton:2018} in order
    to ensure that $\LOPSInt_\EE$ is an internal De Morgan algebra rather than a connection algebra.
  }
\end{defi}

\begin{thmC}[\cite{licata:2018}]
  \label{thm:lops}
  There exists a model of \CTT{} in every LOPS topos.
\end{thmC}

Consider a cosmos $\Mor[F]{\Mode}{\CAT}$ such that each $F\prn{m}$ is a LOPS topos.
Theorem~\ref{thm:lops}
then implies that each mode is a model of cubical type theory, but on its own this is insufficient
to conclude that $F$ assembles into a model of \CubicalMTT{}; we must ensure that each $F\prn{\mu}$
properly preserves interval objects and face lattices. In order to isolate what further properties
we must impose on $F$, we briefly revisit how one interprets constructs in cubical type theory such
as systems and face restrictions in a LOPS topos.

Extending a context $X : \EE$ by an interval variable is given by the product:
$X \times \LOPSInt_\EE$. The
structure of this context extension and of dimension terms more generally follows directly from the
universal property of products along with the De Morgan algebra structure on $\LOPSInt_\EE$; a
dimension term in context $X$ is realized as a morphism $\Mor{X}{\LOPSInt_\EE}$. Similarly, an
element of the face lattice in context $X$ is interpreted as a
morphism. $\Mor{X}{\LOPSFace_\EE}$. Restricting a context by such a face is given by
pullback:\footnote{We have used the familiar set-comprehension notation for restriction by a
  face. Because $\LOPSFace_\EE$ is a subobject of $\Omega$, this coincides with the standard
  interpretation of this notation in a topos.}
\[
  \DiagramSquare{
    height = 1.5cm,
    width = 3cm,
    nw = \Compr{X}{\phi},
    nw/style = pullback ,
    ne = \ObjTerm{},
    sw = X,
    se = \LOPSFace _ {\EE} ,
    east = \top,
    south = \phi ,
  }
\]

Returning to our original question, we can now isolate some of the additional structure required by
a cosmos valued in LOPS topoi to induce a model of \CubicalMTT{}. In particular, a right adjoint
between LOPS topoi will correctly model a dependent right adjoint which appropriately respects
cubical structure when its left adjoint satisfies the following conditions:
\begin{defi}
  A morphism of LOPS topoi is a geometric morphism $\Mor[F_! \Adjoint F]{\EE}{\EE'}$ along with the
  following:
  \begin{enumerate}
  \item An isomorphism of De Morgan algebras
    $\alpha_F : F_!\prn{\LOPSInt_{\EE'}} \cong \LOPSInt_\EE$
  \item A factorization of the canonical map $F_!\prn{\Omega} \to \Omega$ through the cofibration
    classifiers $\beta_F : F_!\prn{\LOPSFace_{\EE'}} \to \LOPSFace_\EE$ such that $\beta_F$ commutes
    with quantification over the interval.
  \end{enumerate}
  The maps $\alpha_F$ and $\beta_F$ are required to be compatible \ie{},
  $\beta_F \circ F_!(-=0) = (-=0) \circ \alpha_F$.
\end{defi}

\begin{rem}
  Whilst $F\prn{\LOPSInt_\EE}$ is a De Morgan algebra, we make no assumption that it be isomorphic
  to $\LOPSInt_{\EE'}$. Doing so would correspond to adding an inverse to \ruleref{int/exc}, but as
  mentioned in \autoref{sec:cub-exchange}, this is not valid in the model of guarded recursion in
  \autoref{sec:examples}; explicity, the right adjoint ``later'' does not preserve the
  interval.
  The case is the same for the object of cofibrations.
\end{rem}

We now have built up the machinery
necessary to define the desired fusion of LOPS topoi and \MTT{} cosmoi:
\begin{defi}
  A cubical \MTT{} cosmos $\Mor[F]{\Mode}{\CAT}$ is an extensional \MTT{} cosmos satisfying the
  following additional restrictions:
  \begin{itemize}
  \item $F \prn m$ is a LOPS topos for each mode $m$,
  \item $F\prn{\mu}$ is a morphism of LOPS topoi for each modality $\mu$,
  \item The interval and face lattice maps are pseudonatural.
  \end{itemize}
\end{defi}

\begin{thm}\label{thm:cubical-MTT-cosmos}
  Any cubical \MTT{} cosmos $F$ induces a model of \CubicalMTT{} with modal context structure
  pseudonaturally equivalent to the pseudofunctor of left adjoints induced by
  Lemma~\ref{lemma:adjoint_pseudofunctor}. If the pseudofunctor of left adjoints is a strict
  2-functor, the modal context structure may instead be chosen to be equal to it.
\end{thm}
\begin{proof}
  A cubical \MTT{} cosmos is in particular an extensional \MTT{} cosmos, meaning that all the rules
  from extensional \MTT{} can be modelled with the strictified 2-functor $\widehat{F}_!$ by
  Theorem~\ref{thm:MTT_cosmos}. Equivalence preserves being a LOPS topos, and thus
  $\widehat{F}_!(m) \simeq F_!(m) = F(m)$ is a LOPS topos, implying we can model all the mode-local
  rules added from \CTT{} (including composition structures for all non-modal types) with
  Theorem~\ref{thm:lops}. It thus remains to construct the exchange principles and composition
  structures on modal types.

  We claim that $\widehat{F}_!$ (or rather, the pseudofunctor of right adjoints $\widehat{F}$
  induced by the dual of Lemma~\ref{lemma:adjoint_pseudofunctor}) also has the cubical components
  of being a cubical \MTT{} cosmoi. We have already argued that $\widehat{F}(m)$ is a LOPS
  topos since it is equivalent to $F(m)$, the fact the naturality squares of these equivalences
  commute up to natural isomorphism is enough to show that $\widehat{F}\prn{\mu}$ is a morphism of
  LOPS topoi, and the pseudonatural coherence of the isomorphisms is preserved since the
  equivalences cohere pseudonaturally.

  As a consequence of this, we have at each mode $m : \Mode$ a De Morgan algebra and a bounded
  distributive lattice $\LOPSInt_{\widehat{F}\prn{m}} , \LOPSFace_{\widehat{F}\prn{m}} :
  \widehat{F}\prn{m}$ and for each modality $\Mor[\mu]{n}{m}$ coherent structure-preserving
  maps $\alpha_{\widehat{F}\prn{\mu}} :
  \widehat{F}_!\prn{\mu}\prn{\LOPSInt_{\widehat{F}\prn{m}}} \cong \LOPSInt_{\widehat{F}\prn{n}}$
  and $\beta_{\widehat{F}\prn{\mu}} : \widehat{F}_!\prn{\mu}\prn{\LOPSFace_{\widehat{F}\prn{m}}}
  \to \LOPSFace_{\widehat{F}\prn{n}}$.

  To define the interval exchange operation, take a dimension term
  $\Mor[r]{\Gamma}{\LOPSInt_{\widehat{F}\prn{m}}}$, and define $\IntArr{\mu,\Gamma}\prn{r}$ as the
  composite:
  \[\begin{tikzpicture}[diagram]
    \node (W) {$\widehat{F}_!\prn{\mu}\prn{\Gamma}$};
    \node [right = 3.6cm of W] (M) {$\widehat{F}_!\prn{\mu}\prn{\LOPSInt_{\widehat{F}\prn{m}}}$};
    \node [right = 3.6cm of M] (E) {$\LOPSInt_{\widehat{F}\prn{n}}$};
    \path [->] (W) edge node [above] {$\widehat{F}_!\prn{\mu}\prn{r}$} (M);
    \path [->] (M) edge node [above] {$\alpha_{\widehat{F}\prn{\mu}}$} (E);
  \end{tikzpicture}\]

  The naturality of $\IntArr{\mu}$ follows from the functoriality of $\widehat{F}_!\prn{\mu}$, and
  they cohere lax naturally since the isomorphisms cohere. To see that it defines morphisms of De
  Morgan algebras, consider the concretely the case of $\land$. Preservation is then the
  commutativity of the following diagram:
  \[\begin{tikzpicture}[diagram]
    \node (SWW) {$\widehat{F}_!\prn{\mu}\prn{\Gamma}$};
    \node [right = 5cm of SWW] (SW)
      {$\widehat{F}_!\prn{\mu}\prn{\LOPSInt_{\widehat{F}\prn{m}}\times\LOPSInt_{\widehat{F}\prn{m}}}$};
    \node [right = 4.5cm of SW] (SE) {$\widehat{F}_!\prn{\mu}\prn{\LOPSInt_{\widehat{F}\prn{m}}}$};
    \node [right = 3cm of SE] (SEE) {$\LOPSInt_{\widehat{F}\prn{n}}$};
    \node [above = of SW] (NW)
      {$\widehat{F}_!\prn{\mu}\prn{\LOPSInt_{\widehat{F}\prn{m}}}\times\widehat{F}_!\prn{\mu}\prn{\LOPSInt_{\widehat{F}\prn{m}}}$};
    \node [above = of SEE] (NEE)
      {$\LOPSInt_{\widehat{F}\prn{n}}\times\LOPSInt_{\widehat{F}\prn{n}}$};
    \path [->] (SWW) edge node [below] {$\widehat{F}_!\prn{\mu}\prn{\Pair{r}{s}}$} (SW);
    \path [->] (SW) edge node [below] {$\widehat{F}_!\prn{\mu}\prn{\land}$} (SE);
    \path [->] (SE) edge node [below] {$\alpha_{\widehat{F}\prn{\mu}}$} (SEE);
    \path [->] (SWW) edge node [left, xshift=-9pt]
      {$\Pair{\widehat{F}_!\prn{\mu}\prn{r}}{\widehat{F}_!\prn{\mu}\prn{s}}$} (NW);
    \path [->] (SW) edge node [right]
      {$\Pair{\widehat{F}_!\prn{\mu}\prn{\pi_1}}{\widehat{F}_!\prn{\mu}\prn{\pi_2}}$} (NW);
    \path [->] (NW) edge node [above]
      {$\alpha_{\widehat{F}\prn{\mu}}\times\alpha_{\widehat{F}\prn{\mu}}$} (NEE);
    \path [->] (NEE) edge node [right] {$\land$} (SEE);
  \end{tikzpicture}\]
  The right rectangle commutes since the $\alpha_{\widehat{F}\prn{\mu}}$ preserves $\land$, and the
  left triangle commutes by the uniqueness of morphisms to products. Preservation of the other
  connectives follow similarly.

  The final thing to verify for intervals is that the uniquely determined dashed arrow in the
  following diagram has an inverse:
  \[\begin{tikzpicture}[diagram]
    \node (SW) {$\widehat{F}_!\prn{\mu}\prn{\Gamma\times\LOPSInt_{\widehat{F}\prn{m}}}$};
    \node [right = 4.5cm of SW] (SM) {$\widehat{F}_!\prn{\mu}\prn{\LOPSInt_{\widehat{F}\prn{m}}}$};
    \node [right = 3cm of SM] (SE) {$\LOPSInt_{\widehat{F}\prn{n}}$};
    \node [above = of SW] (NW) {$\widehat{F}_!\prn{\mu}\prn{\Gamma}$};
    \node [above = of SE] (NE)
      {$\widehat{F}_!\prn{\mu}\prn{\Gamma}\times\LOPSInt_{\widehat{F}\prn{n}}$};
    \path [->] (SW) edge node [below] {$\widehat{F}_!\prn{\mu}\prn{\pi_2}$} (SM);
    \path [->] (SM) edge node [below] {$\alpha_{\widehat{F}\prn{\mu}}$} (SE);
    \path [->] (SW) edge node [left] {$\widehat{F}_!\prn{\mu}\prn{\pi_1}$} (NW);
    \path [->] (NE) edge node [above] {$\pi_1$} (NW);
    \path [->] (NE) edge node [right] {$\pi_2$} (SE);
    \path [dashed,->] (SW) edge (NE);
  \end{tikzpicture}\]
  This follows from $\alpha_{\widehat{F}\prn{\mu}}$ being invertible and $\widehat{F}_!\prn{\mu}$
  preserving finite limits.

  Replaying these arguments for the face lattices completes the construction of the exchange
  principles. In particular, while we have not required it, $\beta_{\widehat{F}\prn{\mu}}$ is always
  homomorphism of distributive lattices, essentially because $\widehat{F}_!\prn{\mu}$ preserves
  monomorphisms and commutes with finite limits and colimits. To construct the interpretations of
  $\ExcSbR$ and $\ExcSbRinv$, we observe that it is a (necessarily unique) map witnessing the
  equivalence of a pair of subobjects over the interpretation of $\LockCxV{\Gamma}$.

  Lastly, we will construct the compositions structures on modal types. For this, we note that the
  model of extensional \MTT{} obtained from Theorem~\ref{thm:MTT_cosmos} supports an inverse
  operation to $\MkBox{-}$ such that every element of a modal type is of the form $\MkBox{a}$.
  Therefore, the equation \ruleref{term-eq/comp-mod} can be taken as-is to fully define a
  composition structure.
\end{proof}

\subsection{Cubical presheaves}

The intended model of cubical type theory is a variant on the standard presheaf mode with types
interpreted as a variant of Kan cubical sets~\cite{cohen:2018}---particular presheaves on the cube
category $\CUBE$ realized as the Lawvere theory of De Morgan algebras. One immediate benefit of the
internal construction of a model of \CTT{} is to generalize this result from cubical sets to
\emph{presheaves valued in cubical sets}~\cite{orton:2018}.  Meanwhile, networks of presheaf
categories connected by the essential geometric morphisms induced by functors between base
categories are known to induce models of \MTT{}~\cite[Section 8]{gratzer:journal:2021}. In fact, a
consequence of Theorem~\ref{thm:cubical-MTT-cosmos} is that these two results can be essentially
combined, thereby giving rise to the most important models of \CubicalMTT{}.

\begin{prop}\label{prop:presheaf-LOPS}
  Let $\CC$ and $\DD$ be small categories, let $\Mor[F]{\CC}{\DD}$ be a functor, and write
  $F ^ \ast$, $F _ !$, and $F _ \ast$ for precomposition respectively left and right Kan extensions
  of $F \times \ArrId {\CUBE}$.
  \begin{enumerate}
  \item The presheaf categories $\PSH {\CC \times \CUBE}$ and $\PSH {\DD \times \CUBE}$ are LOPS
    topoi.
  \item The adjunction $\Pre{F} \Adjoint F_*$ induces a morphism of LOPS topoi.
  \item If $F_!$ is lex the adjunction $F _ ! \Adjoint \Pre{F}$ induces a morphism of LOPS topoi.
  \end{enumerate}
\end{prop}

Before proving the above proposition we recall some standard lemmas.

\begin{lem}\label{lem:lexLan}
  Let $\CC$ and $\DD$ be small categories and $\Mor[F]{\CC}{\DD}$ a functor. Left Kan extension of
  functors sending $\Mor[X]{\CC}{\SET}$ to $\Mor[F_!(X)]{\DD}{\SET}$ is lex iff $\COMMA{F}{d}$ is
  filtered for each $d:\DD$.
\end{lem}

The above result can be alternatively phrased as stating that $\Yo \circ F$ is flat if and only if
$\COMMA{F}{d}$ is filtered for each $d : \DD$. A proof of this standard fact is given by
Borceux~\cite[Proposition 6.1.2]{borceux:vol1:1994}.  We note that this implies in particular that
$\COMMA{F}{d}$ is connected.

\begin{lem}\label{lem:commaproduct}
  Let $\Mor[F]{\CC}{\DD}$ and $\Mor[F']{\CC'}{\DD'}$. Then $\COMMA{F\times F'}{(d,d')}$ is
  equivalent to $\COMMA{F}{d}\times \COMMA{F'}{d'}$ for each $d : \DD$ and $d' : \DD'$.
\end{lem}

\begin{lem}\label{lem:colimTerm}
  Consider a diagram $\Mor[F]{A\times B}{\CC}$ where $B$ has a terminal object $b_1$. Then the
  colimit of $F(a,b)$ over $A\times B$ is isomorphic to the colimit of $F(a, b_1)$ over $A$,
  naturally in $A$.
\end{lem}

We can now prove Proposition~\ref{prop:presheaf-LOPS}:
\begin{proof}
  Note first that $\PSH {\CC \times \CUBE} = \Hom{\Op{\prn {\CC \times \CUBE}}}{\SET} \cong
    \Hom{\Op{\CC}}{\CSET} = \PSH [\CSET] \CC$. Letting $\I , \F : \CSET$ be the interval
  respectively face lattice from~\cite[Section 8.1]{cohen:2018}, we define $\LOPSInt _ \CC \prn
    {c , I} = \I \prn I$ and $\LOPSFace _ \CC \prn {c , I} = \F \prn I$ for $c : \CC$ and
  $I : \CUBE$.

  For (1) these topoi satisfy the Orton-Pitts axioms as noted in~\cite{coquand:2021}. To see that
  the intervals defined above are tiny we proceed as follows: Using the Yoneda lemma along with the
  fact that $\I$ is naturally isomorphic to $\Hom [\CUBE] - {\brc i}$ shows that $\Yo {c , I}
    \times \LOPSInt \cong \Yo {c , I + \brc i}$, and we thus calculate:
  \begin{align*}
    X ^ \LOPSInt \prn {c , I}
    & \cong \Hom{\Yo {c , I}}{X ^ \LOPSInt}\\
    &\cong \Hom{\Yo {c , I} \times \LOPSInt}{X} \\
    & \cong \Hom{\Yo {c , I + \brc i}}{X}\\
    & \cong X \prn {c , I + \brc i} \\
    & \cong \Pre*{\prn{\ArrId {\CC} \times \prn {- + \brc i}}}{X}\prn {c , I}
  \end{align*}
  The above is natural in $X$, and thus exponentiation by $\LOPSInt$ is (naturally isomorphic to)
  the precomposition functor $\Pre{\prn{\ArrId {\CC} \times \prn {- + \brc i}}}$. As this functor
  has a right adjoint, we have shown that $\LOPSInt$ is tiny.

  We write $\pi_{\CC}, \pi_{\DD}$ for the projections $\Mor{\CC \times \CUBE}{\CUBE}$ and
  $\Mor {\DD \times\CUBE}{\CUBE}$ respectively. For the second (respectively third) requirement we
  show the following:
  \begin{itemize}
  \item $\Pre{F}$ (resp.\ $F_!$) preserves finite limits.
  \item $\iota : \Pre{F} \circ \Pre{\pi_\DD} \cong \Pre{\pi_\CC}$
    (resp.\ $F_! \circ \Pre{\pi_\CC} \cong \Pre{\pi_\DD}$)
  \item $\iota$ is an isomorphism of De Morgan algebras at $\I$
    (resp.\ distributive lattices at $\F$).
  \end{itemize}
  The remaining conditions of a morphism of LOPS topoi follow automatically from the naturality
  of $\alpha$ (resp. $\beta$).

  For the first item, we note that $\Pre{F}$ preserves all limits since it is a right adjoint, and
  that $F_!$ preserves finite limits by assumption.

  Next, the desired isomorphism $\Pre{F}(\Pre{\pi_\DD}(X)) \cong \Pre{\pi_\CC}(X)$ can be taken to
  be the identity, justified by the following computation:
  \[
    \Pre{F}(\Pre{\pi_\DD}(X))(c,I) = \Pre{\pi_\DD}(X)(F(c),I) = X(I) = \Pre{\pi_\CC}(X)(c,I)
  \]
  It is clear that this isomorphism preserves the De Morgan algebra structure when $X=\I$ and the
  distributive lattice structure when $X=\F$.

  It remains to consider these conditions for $F_!$. We construct
  $F_!(\Pre{\pi_\CC}(X)) \cong \Pre{\pi_\DD}(X)$ as the composite of a string of natural
  isomorphisms:
  \begin{align*}
    F_!(\Pre{\pi_\CC}(X))(d,I)
    &\cong \Colim_{\prn{\prn{c, I'}, \_} : \COMMA{F \times \ArrId{}}{(d,I)}}
    \Pre{\pi_\CC}(X)(c,I')\\
    &= \Colim_{\prn{\prn{c, I'}, \_} : \COMMA{F \times \ArrId{}}{(d,I)}} X(I')\\
    &\cong \Colim_{\prn{\prn{c,\_},\prn{I',\_}} : \COMMA{F}{d} \times \COMMA{\ArrId{}}{I}} X(I')
    && \text{Lemma~\ref{lem:commaproduct}}\\
    &\cong \Colim_{\prn{c, \_} : \COMMA{F}{d}} X(I)
    && \text{Lemma~\ref{lem:colimTerm}}\\
    &\cong X(I)\\
    &= \Pre{\pi_\DD}(X)(d,I)
  \end{align*}
  In the above calculation, the fourth isomorphism follows by observing that
  $\Colim_{\prn{c, \_} : \COMMA{F}{d}} X(I)$ is the colimit of
  a constant diagram over $\COMMA{F}{d}$; because
  $F_!$ is lex, Lemma~\ref{lem:lexLan} ensures that $\COMMA{\ArrId{\CUBE}\times F}{(d,I)}$ is
  connected, and hence so is $\COMMA{F}{d}$. The isomorphism then follows from the observation that
  colimits of constant, connected diagrams are isomorphic to the value of the diagram.

  We must argue that this is an isomorphism of De Morgan algebras when $X = \I$ and of distributive
  lattices when $X = \F$. Chasing an element through this string of isomorphisms, we send an element
  of the colimit $\In{\prn{\prn{c, I'}, \prn{f,g}}}\prn{x}$ to $X\prn{g}(x)$. One can verify that
  this preserves the relevant structure when $X$ is appropriately specialized. We illustrate the
  simple case of interval endpoints: The $0$ endpoint of $F_!(\LOPSInt_\CC)$ at $(d,I)$ is given by
  $\In{(f_0, \ArrId{})}\prn{0 : \I(I)}$ where $f_0$ is an arbitrary object of the (necessarily
  non-empty) category $\COMMA{F}{d}$. It is clear that this pair is mapped to
  $0 : \LOPSInt_\DD(d,I)$ via the morphisms above.
\end{proof}

We can package all of the above results into the following:

\begin{thm}\label{thm:cubical-presheaf-model}
  Let $\Mor[F]{\Mode}{\CAT}$ be a strict 2-functor, write $\Pre{F}\prn{\mu}$, $\LKan{F}\prn{\mu}$,
  and $\RKan{F}\prn{\mu}$ for the precomposition, left Kan extension, and right Kan extension
  respectively of $F\prn{\mu}\times\ArrId{\CUBE}$, and write $\Pre{F}$, $\LKan{F}$, and $\RKan{F}$
  for the induced pseudofunctors.
  \begin{itemize}
    \item The network of morphisms of LOPS topoi given by the adjunctions
    $\Pre{F}\prn{\mu}\dashv\RKan{F}\prn{\mu}$ induces a model of \CubicalMTT{} over $\Mode$ with
    modal context structure equal to $\Pre{F}$.
    \item The network of morphisms of LOPS topoi given by the adjunctions
    $\LKan{F}\prn{\mu}\dashv\Pre{F}\prn{\mu}$ induces a model of \CubicalMTT{} over $\Coop{\Mode}$
    with modal context structure pseudonaturally equivalent to $\LKan{F}$ if each
    $\LKan{F}\prn{\mu}$ is lex.
  \end{itemize}
\end{thm}
\begin{proof}
  By Theorem~\ref{thm:cubical-MTT-cosmos}, it is sufficient to show that $\RKan{F}$ (respectively
  $\Pre{F}$) is a cubical $\MTT{}$ cosmos. By Proposition~\ref{prop:presheaf-LOPS}, each
  $\PSH{F\prn{m}\times\CUBE}$ is a LOPS topos, and each adjunction $\Pre{F}\prn{\mu} \dashv
  \RKan{F}\prn{\mu}$ (respectively $\LKan{F}\prn{\mu} \dashv \Pre{F}\prn{\mu}$) is a morphism of
  LOPS topoi, and we thus need only verify that the interval and face lattice isomorphisms cohere
  pseudonaturally.

  Consider first the case of the adjunctions $\Pre{F}\prn{\mu} \dashv \RKan{F}\prn{\mu}$. In this
  case, each interval (respectively face lattice) isomorphism is the identity, and thus since the
  left adjoints form a strict 2-functor, the interval objects (respectively face lattice objects)
  form a strict 2-natural transformation, which in particular is also a pseudonatural
  transformation.

  Consider next the case of the adjunctions $\LKan{F}\prn{\mu} \dashv \Pre{F}\prn{\mu}$. To prove
  pseudonaturality, we need to prove coherence with identity, composition, and 2-cells. The proofs
  are similar, and we illustrate them with the identity case. Since $\LKan{F}\prn{\mu}$ is a
  pseudofunctor, there is a natural isomorphism $\LKan{F}\prn{1_m} \cong
  \ArrId{\PSH{F\prn{m}\times\CUBE}}$. We must prove that at the interval (respectively face
  lattice) object, this is the same isomorphism as the one constructed in
  Proposition~\ref{prop:presheaf-LOPS}.

  The isomorphism in the proof of Proposition~\ref{prop:presheaf-LOPS},
  $\LKan{F}\prn{\mu}\prn{\Pre{\pi_{\CC_m}}\prn{X}}\prn{d,I} \cong
  \Pre{\pi_{F\prn{m}}}\prn{X}\prn{d,I}$, may be specified as follows: It is uniquely defined by its
  value upon precomposition with the components of the colimit $\Colim_{\prn{\prn{c, I'}, \_} :
  \COMMA{F\prn{\mu} \times \ArrId{}}{(d,I)}}\Pre{\pi_{F\prn{m}}}\prn{X}\prn{c,I'}$, and this value
  is for each $c : F\prn{m}$, $I' : \CUBE$, and
  $\Mor[\prn{g,\iota}]{\prn{d,I}}{\prn{F\prn{\mu}\prn{c},I'}}$ equal to
  \[
    \Mor[X\prn{\iota}]{\Pre{\pi_{F\prn{m}}}\prn{X}\prn{c,I'} = X\prn{I'}}{X\prn{I} =
    \Pre{\pi_{F\prn{n}}}\prn{X}\prn{d,I}}.
  \]

  The isomorphism $\LKan{F}\prn{1_m} \cong \ArrId{}$ is constructed from the fact that both
  $\LKan{F}\prn{1_m}$ and $\ArrId{}$ are left adjoints to
  $\Pre{F}\prn{1_m} = \ArrId{\PSH{F\prn{m}\times\CUBE}}$. Concretely, it is the counit:
  \[
    \Mor[\varepsilon_{1_m}]{\LKan{F}\prn{1_m}=\LKan{F}\prn{1_m}\circ\Pre{F}\prn{1_m}}{\ArrId{\PSH{F\prn{m}\times\CUBE}}}
  \]
  Precomposing with the components of the colimit
  $\Colim_{\prn{\prn{c, I'}, \_} : \COMMA{F\prn{\mu} \times
      \ArrId{}}{(d,I)}}\Pre{\pi_{F\prn{m}}}\prn{X}\prn{c,I'}$, we get $X\prn{\iota}$ as before,
  which shows that the two isomorphisms are the same, and thus the identity condition for
  pseudonaturality is satisfied.
\end{proof}

\begin{rem}
  In these results, we have chosen the cofibration classifier in each mode to be
  $\LOPSFace _ \CC \prn {c , I} = \F \prn I$. While certainly valid, it is also reasonable to define
  $\LOPSFace$ to simply be the subobject classifier $\Omega$. We leave it to the reader to show that
  with this alternative choice of face lattice, Theorem~\ref{thm:cubical-presheaf-model} is still
  valid.
\end{rem}

\section{Proving and programming with guarded recursion}
\label{sec:examples}

We now turn from theory to practice\footnote{Or at least, slightly more practice-adjacent theory!}
and consider \emph{guarded \CubicalMTT{}}. We briefly recall guarded recursion.  The core idea of
guarded recursion~\cite{nakano:2000} is to use a modality $\Later$ (pronounced `later') to isolate
recursively produced data to prevent its use until work is done, thereby ensuring productivity. This
modality is equipped with operations making it into an applicative functor~\cite{mcbride:2008} which
satisfies \emph{L{\"o}b induction}, a powerful guarded fixed-point principle:
\begin{mathpar}
  \Mor[\Next]{A}{\Later A}
  \and
  \Mor[\prn{\ZApp}]{\Later {\prn {A \to B}}}{\prn{\prn{\Later A} \to \prn{\Later B}}}
  \and
  \Mor[\Lob]{\prn{{\Later A} \to A}}{A}
\end{mathpar}

In particular, $\Lob$ allows us to define an element of $A$ recursively but because the recursively
computed data is available only as $\Later A$, the usual problems with fixed-points are avoided.
We consider a variant of guarded recursion which further includes an idempotent comonad
$\Box$\footnote{This is not related to cubes despite what the notation might suggest.} along
with an equivalence $\Box \Later A \Equiv \Box A$.
This last property ensures that guarded type theory
can construct \emph{coinductive} types through L{\"o}b induction~\cite{clouston:2015}.

To encode guarded recursion in \CubicalMTT{}, we instantiate the theory with a particular mode
theory and extend it with a pair of
axioms. As a result, we obtain a highly workable guarded type theory supporting the relevant
modalities and operations. Similar work was done for \emph{extensional} \MTT{} in \cite[Section
9]{gratzer:journal:2021}; here we show that the improved notion of equality in \CubicalMTT{} results
in an improved experience.

Concretely, we work in the mode theory $\Mode _ g$, a poset-enriched category which is concisely
defined by \autoref{fig:guarded-mode-theory}. Using the substitutions induced by 2-cells, we define:
\begin{mathpar}
  \Later A = \Modify[\ell]{\Sb{A}{\Key{1\leq\ell}{}}}
  \and
  \Next[x] = \MkBox[\ell]{\Sb{A}{\Key{1\leq\ell}{}}}
\end{mathpar}
$\ZApp$ is likewise definable, but $\Lob$ cannot be defined in \CubicalMTT{} and must
be axiomitized (\autoref{fig:lob-induction}).
In order to justify its inclusion, we provide a model of \CubicalMTT{}
over $\Mode _ g$ with L{\"o}b induction.
\begin{figure}
  \[
    \begin{tikzpicture}[diagram]
      \node (T) {$t$};
      \node (S) [right = of T] {$s$};
      \path[->] (T) edge[looseness = 15, out = 210, in = 150] node[left] {$\ell$} (T);
      \path[->] (S) edge[bend left] node[below] {$\delta$} (T);
      \path[->] (T) edge[bend left] node[above] {$\gamma$} (S);
      \node (E) [right = 2.5cm of S]
      {$\begin{aligned}
          \delta \circ \gamma &\le 1 & 1 &= \gamma \circ \delta\\
          1 &\le \ell & \gamma &= \gamma \circ \ell
        \end{aligned}$};
    \end{tikzpicture}
  \]
  \caption{$\Mode_g$: a mode theory for guarded recursion. Reproduced from fig.~11 of
    \cite{gratzer:journal:2021}}
  \label{fig:guarded-mode-theory}
\end{figure}

\begin{figure}
  \begin{mathparpagebreakable}
    \inferrule{
      \IsCx \Gamma <t> \\
      \IsTy A <t>
    }{
      \IsTm \Lob {\prn {\Later A \to A} \to A} <t>
    }
    \and
    \inferrule{
      \IsCx \Gamma <t> \\
      \IsTy A <t> \\
      \IsTm M {\Later A \to A} <t>
    }{
      \EqTm {\App \Lob M} {\App M {\Next [\App \Lob M]}} A <t>
    }
  \end{mathparpagebreakable}
  \caption{The rules of L{\"o}b induction.}
  \label{fig:lob-induction}
\end{figure}

\subsection{Soundness of L{\"o}b induction in \texorpdfstring{\CubicalMTT}{Cubical MTT}}

Letting $\omega$ be the poset category for the first infinite ordinal and $1$ the terminal
category, we define the strict 2-functor $\Mor [f] {\Mode _ g} \CAT$ by
\begin{mathpar}
  f\prn{t} = \omega
  \and
  f\prn{s} = \mathbf 1
  \and
  f\prn{\delta}\prn{\ast} = 0
  \and
  f\prn{\ell}\prn{n} = n+1
  \and
  f\prn{\gamma}\prn{n} = \ast
\end{mathpar}
From this, we define the pseudofunctor $\Mor[F]{\Mode_g}{\CAT}$ by
$F\prn{m} = \PSH{f\prn{m}\times\CUBE}$ and $F\prn{\mu} = \prn{f\prn{\mu}\times\ArrId{\CUBE}}_\ast$,
which by Theorem~\ref{thm:cubical-presheaf-model} induces a model of \CubicalMTT{} $\widehat{F}$.
This
model is almost the same as the model defined in \cite[Section 9.2]{gratzer:journal:2021}, but
$\widehat{F}$ uses $\CSET$-valued presheaves. Since the cubical and modal aspects of \CubicalMTT{}
are orthogonal, considerations in the $\SET$-based model that do not involve identity types carry
over to $\widehat{F}$. In particular, because L{\"o}b induction holds in the $\SET$-based model,
it also holds in $\widehat{F}$.\footnote{This can also be verified by hand as is done in
  \cite{birkedal:2019}.}

\subsection{Programming with guarded \texorpdfstring{\CubicalMTT}{Cubical MTT}}

To see that \CubicalMTT{} can not merely replicate but also improve on work done in \MTT{}, we
now show that L{\"o}b induction not only gives a fixpoint but a unique one. In \cite[Theorem
9.5]{gratzer:2020} this is proven for extensional \MTT{} (by introducing equality reflection), but
because of
modal extensionality, we can now prove it with nothing but \CubicalMTT{} and L{\"o}b induction.
Similarly, the results from \cite[Section 9.4]{gratzer:journal:2021} about guarded and coinductive
streams in guarded \MTT{} may also be proven in guarded \CubicalMTT{} without equality reflection.

\begin{thm}\label{thm:lob-unique-fixpoint}
$\App \Lob M$ is the unique guarded fixpoint of $\Mor [M] {\Later {\Dec A}} {\Dec A}$, i.e.
\[
  \prn {A : \Uni}\prn{x : {\Dec A}}
  \to \Path{\Dec A}{\App M {\Next[x]}}{x}
  \to \Path {\Dec A}{\Lob[M]}{x}
\]
\end{thm}

\begin{proof}
  Supposing $A : \Uni$, we intend to use L{\"o}b induction to find a term of
  \[
    \prn{x : {\Dec A}} \to \Path{\Dec A}{\App M {\Next[x]}}{x} \to \Path {\Dec A}{\Lob[M]}{x}
  \]
  To this end, given $f : \Later \prn{\prn{x : {\Dec A}} \to \Path{\Dec A}{\App M
  {\Next[x]}}{x} \to \Path {\Dec A}{\Lob[M]}{x}}$,
  $x : \Dec A$, and $p : \Path {\Dec A} {\App M {\Next [x]}} x$, we must define a term of
  $\Path {\Dec A} {\App \Lob M} x$. We can construct the term
  \[
    f \ZApp \Next [x] \ZApp \Next [p] : \Later {\Path {\Dec A} {\Lob[M]} x}.
  \]
  By Theorem~\ref{thm:cmtt:modal-ext}, this gives a term of
  $\Path {\Later {\Dec A}} {\Next [\Lob[M]]} {\Next [x]}$. Using that function application
  preserves paths
  and that $\App \Lob M$ is a guarded fixpoint we then obtain the paths
  \[
    \App \Lob M = \App M {\Next [\Lob M]}  \PathEq \App M {\Next [x]}  \PathEq x. \qedhere
  \]
\end{proof}

\section{Related work}
\label{sec:related-work}

\CubicalMTT{} builds upon two distinct strands of work: cubical and modal type theories. Even though
both lines of research are ongoing, several proposals have already been made which combine
elements of both.

\paragraph*{Modal homotopy type theory}
Several version of homotopy type theory extended with modalities have been
proposed~\cite{shulman:2018,riley:2021}. These type theories aim to increase the expressivity of
HoTT and allow it to better capture some aspects of homotopy theory. Unlike \CubicalMTT{}, however,
these theories tend to be specialized to various modal situations. They build in the structure of
one specific modality and are hand-crafted to have manageable syntax for that situation. In
contrast, \CubicalMTT{} follows \MTT{} and works for a \emph{class} of modalities, and provides
usable syntax in each case. Moreover, prior type theories in this tradition expand ``book
HoTT''~\cite{hottbook} and therefore do not enjoy the good computational properties we conjecture
for
\CubicalMTT{}.

\paragraph*{Modal cubical type theory}
In order to extend cubical type theory with an \emph{internal} notion of parametricity,
Cavallo~\cite[Part IV]{cavallo:phd} has proposed a variant of (cartesian) cubical type theory
extended with connectives and a handful of modalities to capture parametricity. Like \CubicalMTT{},
this \emph{cohesive parametric cubical type theory} combines cubical type theory with Fitch-style
modalities. While morally the system is a specialization of \CubicalMTT{} to a cohesive collection
of modalities, Cavallo takes advantage of several specifics of the intended model to add various
equations to the theory.

Separately, another Fitch-style type theory, \emph{clocked type theory}, has been extended to a
cubical basis~\cite{kristensen:2021}. This theory is used to present guarded recursion, similarly to
\autoref{sec:examples}. Unlike guarded \CubicalMTT{}, clocked cubical type theory includes several
specialized axioms, a more sophisticated collection of guarded modalities, and an account of the
interaction of HITs with parts of the modal machinery.

The extra equations and properties of the modalities in both systems prevent \CubicalMTT{} from
directly recovering either parametric cubical type theory or clocked cubical type theory. The core
aspects of both, however, are similar to \CubicalMTT{} and we believe that \CubicalMTT{} gives a
means of systematically generalizing these calculi to other modal situations.

\section{Conclusions}
\label{sec:conclusions}

We contribute \CubicalMTT{}, a general modal type theory based on cubical type theory and
\MTT{}. The system can be instantiated to a number of modal situations while still maintaining
computationally effective interpretations of univalence and function extensionality.

While in this work we have introduced the theory and characterized a class of models for it, in the
future we hope to investigate further metatheoretic properties of the system. In particular, both
\MTT{} and cubical type theory enjoy normalization~\cite{gratzer:normalization:2022,sterling:2021}, and we
conjecture that these proofs can be combined and generalized to apply to \CubicalMTT{}. The
introduction of cubical cosmoi takes the first step in this direction: cosmoi are a crucial
ingredient of the proof of normalization for \MTT{}. In a separate direction, we hope to investigate
the behavior of more of the mode-local structure of cubical type theory such as \emph{higher
  inductive types} and other novelties of cubical type theories.

\bibliographystyle{alphaurl}
\bibliography{refs.bib,temp-refs.bib}

\appendix
\section{Rules of \texorpdfstring{\CubicalMTT{}}{Cubical MTT}}
\label{app:rules}

We here present the official syntax and rules of \CubicalMTT{}. For the sake of brevity, we omit a
number the rules, especially those lifted from \MTT{} or \CTT{}; in particular, we omit the
following:
\begin{itemize}
\item The rules for dependent sums, booleans, universes.
\item The equations stating that the interval is a De Morgan algebra.
\item The equations stating that $\ExcIntInv{\prn{-}}$ for interval terms is a morphism of De Morgan
  algebras.
\item The equations stating that the face lattice is a bounded distributive lattice,
\item The equations stating that $\ExcCofInv{\prn{-}}$ for faces is a morphism of bounded lattices,
\item Miscellaneous equations commuting substitutions past term formers or governing the composition
  of substitutions.
\end{itemize}
Rules that govern the interaction between cubical and modal aspects are marked with a ${ }^\dagger$.
At the end there is a section on derived definitions some of which we will use throughout to ease
notation.

\paragraph*{Context formation.}\hfill

\begin{mathparpagebreakable}
  \inferrule [cx/emp]
  {\phantom{.}}
  {\IsCx \EmpCx} \and
  \inferrule [cx/lock]
  {\Mor [\mu] n m \\
    \IsCx \Gamma}
  {\IsCx {\LockCx \Gamma} <n>} \and
  \inferrule [cx/ext-type]
  {\Mor [\mu] n m \\
    \IsCx \Gamma \\
    \IsTy [\LockCx \Gamma] A <n>}
  {\IsCx {\ECx \Gamma A}} \and
  \inferrule [cx/ext-int]
  {\IsCx \Gamma}
  {\IsCx {\ICx \Gamma}} \and
  \inferrule [cx/face-res]
  {\IsCx \Gamma \\
    \IsTm \phi {\F _ m}}
  {\IsCx {\RCx \Gamma}} \and
\end{mathparpagebreakable}

\paragraph*{Context equality.}\hfill

\begin{mathparpagebreakable}
  \inferrule [cx-eq/comp-lock]
  {\Mor [\mu] n m \\
    \Mor [\nu] o n \\
    \IsCx \Gamma}
  {\EqCx {\LockCx \Gamma <\mu \circ \nu>} {\LockCx {\LockCx \Gamma} <\nu>} <o>}
  \and
  \inferrule [cx-eq/id-lock]
  {\IsCx \Gamma}
  {\EqCx {\LockCx \Gamma <1>} \Gamma} \and
\end{mathparpagebreakable}

\paragraph*{Substitution formation.}\hfill

\begin{mathparpagebreakable}
  \inferrule [sb/comp]
  {\IsCx {\Gamma , \Delta , \Xi} \\
    \IsSb \delta \Delta \\
    \IsSb [\Delta] \xi \Xi}
  {\IsSb {\xi \circ \delta} \Xi} \and
  \inferrule [sb/id]
  {\IsCx \Gamma}
  {\IsSb \ISb \Gamma} \and
  \inferrule [sb/emp]
  {\IsCx \Gamma}
  {\IsSb \EmpSb \EmpCx} \and
  \inferrule [sb/weak-type]
  {\Mor [\mu] n m \\
    \IsCx \Gamma \\
    \IsTy [\LockCx \Gamma] A <n>}
  {\IsSb [\ECx \Gamma A] \Wk \Gamma} \and
  \inferrule [sb/weak-int]
  {\IsCx \Gamma}
  {\IsSb [\ICx \Gamma] \WkI \Gamma} \and
  \inferrule [sb/weak-res]
  {\IsCx \Gamma \\
    \IsTm \phi {\F _ m}}
  {\IsSb [\RCx \Gamma] \WkF \Gamma} \and
  \inferrule [sb/lock]
  {\Mor [\mu] n m \\
    \IsCx {\Gamma , \Delta} \\
    \IsSb \delta \Delta}
  {\IsSb [\LockCx \Gamma] {\LockSb \delta} {\LockCx \Delta} <n>} \and
  \inferrule [sb/key]
  {\mu , \Mor [\nu] n m \\
    \Mor [\alpha] \nu \mu \\
    \IsCx \Gamma}
  {\IsSb [\LockCx \Gamma] {\Key \alpha \Gamma} {\LockCx \Gamma <\nu>} <n>} \and
  \inferrule [sb/ext-type]
  {\Mor [\mu] n m \\
    \IsCx {\Gamma , \Delta} \\
    \IsSb \delta \Delta \\
    \IsTy [\LockCx \Delta] A <n> \\
    \IsTm [\LockCx \Gamma] a {\Sb A {\LockSb \delta}} <n>}
  {\IsSb {\ESb \delta a} {\ECx \Delta A}} \and
  \inferrule [sb/ext-int]
  {\IsCx {\Gamma , \Delta} \\
    \IsSb \delta \Delta \\
    \IsTm r {\I _ m}}
  {\IsSb {\ESb \delta r} {\ICx \Delta}} \and
  \inferrule [sb/face-res]
  {\IsCx {\Gamma , \Delta} \\
    \IsSb \delta \Delta \\
    \IsTm [\Delta] \phi {\F _ m} \\
    \EqTm {\Sb \phi \delta} \Ftop {\F _ m}}
  {\IsSb {\RCx \delta} {\RCx \Delta}} \and
  \inferrule [sb/exc-int-inv${ }^\dagger$]
  {\Mor [\mu] n m \\
    \IsCx \Gamma}
  {\IsSb [\ICx {\LockCx \Gamma} <n>] \ExcSbIinv {\LockCx {\ICx \Gamma}} <n>} \and
  \inferrule [sb/exc-face-inv${ }^\dagger$]
  {\Mor [\mu] n m \\
    \IsCx \Gamma \\
    \IsTm [\Gamma] \phi {\F _ m}}
  {\IsSb [\RCx {\LockCx \Gamma} [\ExcCofInv \phi]] \ExcSbRinv {\LockCx {\RCx \Gamma}} <n>}
\end{mathparpagebreakable}

\paragraph*{Substitution equality.}\hfill

\begin{mathparpagebreakable}
  \inferrule [sb-eq/comp-lock]
  {\Mor [\mu] n m \\
    \Mor [\nu] o n \\
    \IsCx {\Gamma , \Delta} \\
    \IsSb \delta \Delta}
  {\EqSb [\LockCx \Gamma <\mu \circ \nu>] {\LockSb \delta <\mu \circ \nu>} {\LockSb {\LockSb \delta
  <\mu>} <\nu>} {\LockCx \Delta <\mu \circ \nu>} <o>} \and
  \inferrule [sb-eq/id-lock]
  {\IsCx {\Gamma , \Delta} \\
    \IsSb \delta \Delta}
  {\EqSb {\LockSb \delta <1>} \delta \Delta} \and
  \inferrule [sb-eq/lock-comp]
  {\Mor [\mu] n m \\
    \IsCx {\Gamma , \Delta , \Xi} \\
    \IsSb \delta \Delta \\
    \IsSb [\Delta] \xi \Xi}
  {\EqSb [\LockCx \Gamma] {\LockSb {\prn {\xi \circ \delta}}} {\LockSb \xi \circ
      \LockSb \delta} {\LockCx \Xi} <n>} \and
  \inferrule [sb-eq/lock-id]
  {\Mor [\mu] n m \\
    \IsCx \Gamma}
  {\EqSb [\LockCx \Gamma] {\LockSb \ISb} \ISb {\LockCx \Gamma} <n>} \and
  \inferrule [sb-eq/id-key]
  {\Mor [\mu] n m \\
    \IsCx \Gamma}
  {\EqSb [\LockCx \Gamma] {\Key {1_\mu} \Gamma} \ISb {\LockCx \Gamma} <n>} \and
  \inferrule [sb-eq/nat-key]
  {\mu , \Mor [\nu] n m \\
    \Mor [\alpha] \nu \mu \\
    \IsCx {\Gamma , \Delta} \\
    \IsSb \delta \Delta}
  {\EqSb [\LockCx \Gamma] {\Key \alpha \Delta \circ \LockSb \delta} {\LockSb \delta
      <\nu> \circ \Key \alpha \Gamma} {\LockCx \Delta <\nu>} <n>} \and
  \inferrule [sb-eq/comp-key]
  {\Mor [\mu , \nu , \rho] n m \\
    \Mor [\alpha] \nu \mu \\
    \Mor [\beta] \rho \nu \\
    \IsCx \Gamma}
  {\EqSb [\LockCx \Gamma] {\Key {\alpha \circ \beta} \Gamma} {\Key \alpha \Gamma \circ
      \Key \beta \Gamma} {\LockCx \Gamma <\rho>} <n>} \and
  \inferrule [sb-eq/whisk-key]
  {\Mor [\mu _ 0 , \mu _ 1] n m \\
    \Mor [\nu _ 0 , \nu _ 1] o n \\
    \Mor [\alpha] {\mu _ 1} {\mu _ 0} \\
    \Mor [\beta] {\nu _ 1} {\nu _ 0} \\
    \IsCx \Gamma}
  {\EqSb [\LockCx \Gamma <\mu _ 0 \circ \nu _ 0>] {\Key {\alpha \Whisker \beta} \Gamma}
    {\LockSb {\Key \alpha \Gamma} <\nu _ 1> \circ \Key \beta {\LockCx \Gamma <\mu _ 0>}}
    {\LockCx \Gamma <\mu _ 1 \circ \nu _ 1>} <o>} \and
  \inferrule [sb-eq/ext-type-beta]
  {\IsCx {\Gamma , \Delta} \\
    \IsSb \delta \Delta \\
    \IsTy [\LockCx \Delta] A <n> \\
    \IsTm [\LockCx \Gamma] a {\Sb A {\LockSb \delta}} <n>}
  {\EqSb {\Wk \circ \ESb \delta a} \delta \Delta} \and
  \inferrule [sb-eq/ext-type-eta]
  {\IsCx {\Gamma , \Delta} \\
    \IsTy [\LockCx \Delta] A <n> \\
    \IsSb \delta {\ECx \Delta A}}
  {\EqSb \delta {\ESb {\prn {\Wk \circ \delta}} {\Sb {\Var {0}} \delta}} {\ECx \Delta A}} \and
  \inferrule [sb-eq/ext-int-beta]
  {\IsCx {\Gamma , \Delta} \\
    \IsSb \delta \Delta \\
    \IsTm r {\I _ m}}
  {\EqSb {\WkI \circ \ESb \delta r} \delta \Delta} \and
  \inferrule [sb-eq/ext-int-eta]
  {\IsCx {\Gamma , \Delta} \\
    \IsSb \delta {\ICx \Delta}}
  {\EqSb \delta {\ESb {\prn {\WkI \circ \delta}} {\Sb \VarI \delta}} {\ICx \Delta}} \and
  \inferrule [sb-eq/exc-int-left-inv${ }^\dagger$]
  {\Mor [\mu] n m \\
    \IsCx \Gamma}
  {\EqSb [\LockCx {\ICx \Gamma}] {\ExcSbIinv \circ \ExcSbI} \ISb {\LockCx {\ICx \Gamma}} <n>} \and
  \inferrule [sb-eq/exc-int-right-inv${ }^\dagger$]
  {\Mor [\mu] n m \\
    \IsCx \Gamma}
  {\EqSb [\ICx {\LockCx \Gamma}] {\ExcSbI \circ \ExcSbIinv} \ISb {\ICx {\LockCx \Gamma}} <n>} \and
  \inferrule [sb-eq/exc-face-left-inv${ }^\dagger$]
  {\Mor [\mu] n m \\
    \IsCx \Gamma \\
    \IsTm \phi {\F _ m}}
  {\EqSb [\LockCx {\RCx \Gamma}] {\ExcSbRinv \circ \ExcSbR} \ISb {\LockCx {\RCx \Gamma}} <n>} \and
  \inferrule [sb-eq/exc-face-right-inv${ }^\dagger$]
  {\Mor [\mu] n m \\
    \IsCx \Gamma \\
    \IsTm \phi {\F _ m}}
  {\EqSb [\RCx {\LockCx \Gamma} [\ExcCofInv \phi]] {\ExcSbR \circ \ExcSbRinv} \ISb {\RCx {\LockCx
  \Gamma} [\ExcCofInv \phi]} <n>} \and
  \inferrule [sb-eq/face-res-uniq]
  {\IsCx {\Gamma , \Delta} \\
    \IsTm [\Delta] \phi {\F _ m} \\
    \IsSb \delta {\RCx \Gamma}}
  {\EqSb \delta {\RCx {\prn {\WkF \circ \delta}}} {\RCx \Gamma}} \and
  \inferrule [sb-eq/face-res-bin]
  {\IsCx {\Gamma , \Delta} \\
    \IsSb {\delta , \xi} \Delta \\
    \IsTm {\phi , \psi} {\F _ m} \\
    \EqTm {\phi \lor \psi} \Ftop {\F _ m} \\
    \EqSb [\RCx \Gamma] {\delta \circ \WkF} {\xi \circ \WkF} \Delta \\
    \EqSb [\RCx \Gamma [\psi]] {\delta \circ \WkF [\psi]} {\xi \circ \WkF [\psi]} \Delta}
  {\EqSb \delta \xi \Delta} \and
  \inferrule [sb-eq/face-res-null]
  {\IsCx {\Gamma , \Delta} \\
    \IsSb {\delta , \xi} \Delta \\
    \EqTm \Fbot \Ftop {\F _ m}}
  {\EqSb \delta \xi \Delta}
\end{mathparpagebreakable}

\paragraph*{Interval formation.}\hfill

\begin{mathparpagebreakable}
  \inferrule [int/join]
  {\IsCx \Gamma \\
    \IsTm {r , s} {\I _ m}}
  {\IsTm {r \lor s} {\I _ m}} \and
  \inferrule [int/meet]
  {\IsCx \Gamma \\
    \IsTm {r , s} {\I _ m}}
  {\IsTm {r \land s} {\I _ m}} \and
  \inferrule [int/bot]
  {\IsCx \Gamma}
  {\IsTm 0 {\I _ m}} \and
  \inferrule [int/top]
  {\IsCx \Gamma}
  {\IsTm 1 {\I _ m}} \and
  \inferrule [int/inv]
  {\IsCx \Gamma \\
    \IsTm r {\I _ m}}
  {\IsTm {\Iinv [r]} {\I _ m}} \and
  \inferrule [int/exc${ }^\dagger$]
  {\Mor [\mu] n m \\
    \IsCx \Gamma \\
    \IsTm r {\I _ m}}
  {\IsTm [\LockCx \Gamma] {\ExcIntInv r} {\I _ n} <n>} \and
  \inferrule [int/var]
  {\IsCx \Gamma}
  {\IsTm [\ICx \Gamma] \VarI {\I _ m}} \and
  \inferrule [int/sb]
  {\IsCx {\Gamma , \Delta} \\
    \IsSb \delta \Delta \\
    \IsTm [\Delta] r {\I _ m}}
  {\IsTm {\Sb r \delta} {\I _ m}}
\end{mathparpagebreakable}

\paragraph*{Interval equality.}\hfill

\begin{mathparpagebreakable}
  \inferrule [int-eq/ext-int-beta]
  {\IsCx {\Gamma , \Delta} \\
    \IsSb \delta \Delta \\
    \IsTm r {\I _ m}}
  {\EqTm {\Sb \VarI {\ESb \delta r}} r {\I _ m}} \and
  \inferrule [int-eq/res-eq]
  {\IsCx \Gamma \\
    \IsTm r {\I _ m}}
  {\EqTm [\RCx \Gamma [\EqBot r]] {\Sb r {\WkF [\EqBot r]}} 0 {\I _ m}} \and
  \inferrule [int-eq/exc-comp${ }^\dagger$]
  {\Mor [\mu] n m \\
    \Mor [\nu] o n \\
    \IsCx \Gamma \\
    \IsTm r {\I _ m}}
  {\EqTm [\LockCx {\LockCx \Gamma} <\nu>] {\ExcIntInv r <\mu \circ \nu>} {\ExcIntInv
  {\prn{\ExcIntInv r}} <\nu>} {\I _ o} <o>} \and
  \inferrule [int-eq/exc-id${ }^\dagger$]
  {\IsCx \Gamma \\
    \IsTm r {\I _ m}}
  {\EqTm {\ExcIntInv r <1>} r {\I _ m}} \and
  \inferrule [int-eq/exc-key${ }^\dagger$]
  {\Mor [\mu , \nu] n m \\
    \Mor [\alpha] \nu \mu \\
    \IsCx \Gamma \\
    \IsTm r {\I _ m}}
  {\EqTm [\LockCx \Gamma] {\Sb {\ExcIntInv r <\nu>} {\Key \alpha \Gamma}} {\ExcIntInv r} {\I _ n}
  <n>} \and
  \inferrule [int-eq/exc-sub${ }^\dagger$]
  {\Mor [\nu] n m \\
    \IsCx {\Gamma , \Delta} \\
    \IsSb \delta \Delta \\
    \IsTm r {\I _ m}}
  {\EqTm [\LockCx \Gamma] {\Sb {\ExcIntInv r} {\LockSb \delta}} {\ExcIntInv {\Sb r \delta}} {\I _
  n} <n>} \and
  \inferrule [int-eq/face-res-bin]
  {\IsCx \Gamma \\
    \IsTm {r , s} {\I _ m} \\
    \IsTm {\phi , \psi} {\F _ m} \\
    \EqTm {\phi \lor \psi} \Ftop {\F _ m} \\
    \EqTm [\RCx \Gamma] {\Sb r \WkF} {\Sb s \WkF} {\I _ m} \\
    \EqTm [\RCx \Gamma [\psi]] {\Sb r {\WkF [\psi]}} {\Sb s {\WkF [\psi]}} {\I _ m}}
  {\EqTm r s {\I _ m}} \and
  \inferrule [int-eq/face-res-null]
  {\IsCx \Gamma \\
    \IsTm {r , s} {\I _ m} \\
    \EqTm \Fbot \Ftop {\F _ m}}
  {\EqTm r s {\I _ m}}
\end{mathparpagebreakable}

\paragraph*{Face formation.}\hfill

\begin{mathparpagebreakable}
  \inferrule [face/eq]
  {\IsCx \Gamma \\
    \IsTm r {\I _ m}}
  {\IsTm {\EqBot r} {\F _ m}} \and
  \inferrule [face/join]
  {\IsCx \Gamma \\
    \IsTm {\phi , \psi} {\F _ m}}
  {\IsTm {\phi \lor \psi} {\F _ m}} \and
  \inferrule [face/meet]
  {\IsCx \Gamma \\
    \IsTm {\phi , \psi} {\F _ m}}
  {\IsTm {\phi \land \psi} {\F _ m}} \and
  \inferrule [face/bot]
  {\IsCx \Gamma}
  {\IsTm \Fbot {\F _ m}} \and
  \inferrule [face/top]
  {\IsCx \Gamma}
  {\IsTm \Ftop {\F _ m}} \and
  \inferrule [face/exc${ }^\dagger$]
  {\Mor [\mu] n m \\
    \IsCx \Gamma \\
    \IsTm \phi {\F _ m}}
  {\IsTm [\LockCx \Gamma] {\ExcCofInv \phi} {\F _ n} <n>} \and
  \inferrule [face/sb]
  {\IsCx {\Gamma , \Delta} \\
    \IsSb \delta \Delta \\
    \IsTm [\Delta] \phi {\F _ m}}
  {\IsTm {\Sb \phi \delta} {\F _ m}} \and
\end{mathparpagebreakable}

\paragraph*{Face equality.}\hfill

\begin{mathparpagebreakable}
  \inferrule [face-eq/non-contr]
  {\IsCx \Gamma \\
    \IsTm r {\I _ m}}
  {\EqTm {\EqBot r \land \EqBot {\Iinv [r]}} \Fbot {\F _ m}} \and
  \inferrule [face-eq/exc-comp${ }^\dagger$]
  {\Mor [\mu] n m \\
    \Mor [\nu] o n \\
    \IsCx \Gamma \\
    \IsTm \phi {\F _ m} <m>}
  {\EqTm [\LockCx {\LockCx \Gamma} <\nu>] {\ExcCofInv \phi <\mu \circ \nu>} {\ExcCofInv {\prn
  {\ExcCofInv \phi}} <\nu>} {\F _ o} <o>} \and
  \inferrule [face-eq/exc-id${ }^\dagger$]
  {\IsCx \Gamma \\
    \IsTm \phi {\F _ m}}
  {\EqTm {\ExcCofInv \phi <1>} \phi {\F _ m}} \and
  \inferrule [face-eq/exc-key${ }^\dagger$]
  {\mu , \Mor [\nu] n m \\
    \Mor [\alpha] \nu \mu \\
    \IsCx \Gamma \\
    \IsTm \phi {\F _ m} <m>}
  {\EqTm [\LockCx \Gamma] {\Sb {\ExcCofInv \phi <\nu>} {\Key \alpha \Gamma}} {\ExcCofInv \phi} {\F
  _ n} <n>} \and
  \inferrule [face-eq/exc-eq${ }^\dagger$]
  {\Mor [\mu] n m \\
    \IsCx \Gamma \\
    \IsTm r {\I _ m} <m>}
  {\EqTm {\ExcIntInv {\EqBot r}} {\EqBot {\ExcIntInv r}} {\F _ n} <n>} \and
  \inferrule [face-eq/exc-sub${ }^\dagger$]
  {\Mor [\mu] n m \\
    \IsCx {\Gamma , \Delta} \\
    \IsSb \delta \Delta \\
    \IsTm \phi {\F _ m}}
  {\EqTm [\LockCx \Gamma] {\Sb {\ExcCofInv \phi} {\LockSb \delta}} {\ExcCofInv {\Sb \phi \delta}}
  {\F _ n} <n>} \and
  \inferrule [face-eq/res-eq-top]
  {\IsCx \Gamma \\
    \IsTm \phi {\F _ m}}
  {\EqTm [\RCx \Gamma] {\Sb \phi \WkF} \Ftop {\F _ m}} \and
  \inferrule [face-eq/eq-zero]
  {\IsCx \Gamma}
  {\EqTm {\EqBot 0} \Ftop {\F _ m}} \and
  \inferrule [face-eq/face-res-bin]
  {\IsCx \Gamma \\
    \IsTm {\phi , \psi , \chi _ 0 , \chi _ 1} {\F _ m} \\
    \EqTm {\phi \lor \psi} \Ftop {\F _ m} \\
    \EqTm [\RCx \Gamma] {\Sb {\chi _ 0} \WkF} {\Sb {\chi _ 1} \WkF} {\F _ m} \\
    \EqTm [\RCx \Gamma [\psi]] {\Sb {\chi _ 0} {\WkF [\psi]}} {\Sb {\chi _ 1} {\WkF [\psi]}} {\F _
    m}}
  {\EqTm {\chi _ 0} {\chi _ 1} {\F _ m}} \and
  \inferrule [face-eq/face-res-null]
  {\IsCx \Gamma \\
    \IsTm {\chi _ 0 , \chi _ 1} {\F _ m} \\
    \EqTm \Fbot \Ftop {\F _ m}}
  {\EqTm {\chi _ 0} {\chi _ 1} {\F _ m}}
\end{mathparpagebreakable}

\paragraph*{Type formation.}\hfill

\begin{mathparpagebreakable}
  \inferrule [type/pi]
  {\Mor [\mu] n m \\
    \IsCx \Gamma \\
    \IsTy [\LockCx \Gamma] A <n> \\
    \IsTy [\ECx \Gamma A] B}
  {\IsTy {\Fn A B}} \and
  \inferrule [type/path]
  {\IsCx \Gamma \\
    \IsTy [\ICx \Gamma] A \\
    \IsTm a {\Sb A {\ESb \ISb 0}} \\
    \IsTm b {\Sb A {\ESb \ISb 1}}}
  {\IsTy {\Path A a b}} \and
  \inferrule [type/mod]
  {\Mor [\mu] n m \\
    \IsCx \Gamma \\
    \IsTy [\LockCx \Gamma] A <n>}
  {\IsTy {\Modify {A}}} \and
  \inferrule [type/sys]
  {\IsCx \Gamma \\
    \IsTm {\phi , \psi} {\F _ m} \\
    \EqTm {\phi \lor \psi} \Ftop {\F _ m} \\
    \IsTy [\RCx \Gamma] A \\
    \IsTy [\RCx \Gamma [\psi]] B \\
    \EqTy [\RCx \Gamma [\phi \land \psi]] {\Sb A {\RCx {\WkF [\phi \land \psi]}}}
    {\Sb B {\RCx {\WkF [\phi \land \psi]} [\psi]}}}
  {\IsTy {\SysBin \phi A \psi B}} \and
  \inferrule [type/sb]
  {\IsCx {\Gamma , \Delta} \\
    \IsSb \delta \Delta \\
    \IsTy [\Delta] A}
  {\IsTy {\Sb A \delta}}
\end{mathparpagebreakable}

\pagebreak

\paragraph*{Type equality.}\hfill

\begin{mathparpagebreakable}
  \inferrule [type-eq/sys-top]
  {\IsCx \Gamma \\
    \IsTm \phi {\F _ m} \\
    \IsTy [\RCx \Gamma [\Ftop]] A \\
    \IsTy [\RCx \Gamma] B \\
    \EqTy [\RCx \Gamma] {\Sb A {\RCx \WkF [\Ftop]}} B}
  {\EqTy {\SysBin \Ftop A \phi B} {\Sb A {\RCx \ISb [\Ftop]}}} \and
  \inferrule [type-eq/face-res-bin]
  {\IsCx \Gamma \\
    \IsTm {\phi , \psi} {\F _ m} \\
    \EqTm {\phi \lor \psi} \Ftop {\F _ m} \\
    \IsTy {A , B} \\
    \EqTy [\RCx \Gamma] {\Sb A \WkF} {\Sb B \WkF} \\
    \EqTy [\RCx \Gamma [\psi]] {\Sb A {\WkF [\psi]}} {\Sb B {\WkF [\psi]}}}
  {\EqTy A B} \and
  \inferrule [type-eq/face-res-null]
  {\IsCx \Gamma \\
    \EqTm \Fbot \Ftop {\F _ m} \\
    \IsTy {A , B}}
  {\EqTy A B}
\end{mathparpagebreakable}

\paragraph*{Term formation.}\hfill

\begin{mathparpagebreakable}
  \inferrule [term/pi-lam]
  {\Mor [\mu] n m \\
    \IsCx \Gamma \\
    \IsTy [\LockCx \Gamma] A <n> \\
    \IsTy [\ECx \Gamma A] B \\
    \IsTm [\ECx \Gamma A] b B}
  {\IsTm {\Lam b} {\Fn A B}} \and
  \inferrule [term/pi-app]
  {\Mor [\mu] n m \\
    \IsCx \Gamma \\
    \IsTy [\LockCx \Gamma] A <n> \\
    \IsTy [\ECx \Gamma A] B \\
    \IsTm f {\Fn A B} \\
    \IsTm [\LockCx \Gamma] a A <n>}
  {\IsTm {\App f a} {\Sb B {\ESb \ISb a}}} \and
  \inferrule [term/path-abs]
  {\IsCx \Gamma \\
    \IsTy [\ICx \Gamma] A \\
    \IsTm [\ICx \Gamma] a A}
  {\IsTm {\PathAbs a} {\Path A {\Sb a {\ESb \ISb 0}} {\Sb a {\ESb \ISb 1}}}} \and
  \inferrule [term/path-app]
  {\IsCx \Gamma \\
    \IsTy [\ICx \Gamma] A \\
    \IsTm a {\Sb A {\ESb \ISb 0}} \\
    \IsTm b {\Sb A {\ESb \ISb 1}} \\
    \IsTm p {\Path A a b} \\
    \IsTm r {\I _ m}}
  {\IsTm {\PathApp p r} {\Sb A {\ESb \ISb r}}} \and
  \inferrule [term/mod-mod]
  {\Mor [\mu] n m \\
    \IsCx \Gamma \\
    \IsTy [\LockCx \Gamma] A <n>\\
    \IsTm [\LockCx \Gamma] a A <n>}
  {\IsTm {\MkBox a} {\Modify {A}}} \and
  \inferrule [term/mod-let]
  {\Mor [\mu] n m \\
    \Mor [\nu] o n \\
    \IsCx \Gamma \\
    \IsTy [\LockCx {\LockCx \Gamma} <\nu>] A <o> \\
    \IsTm [\LockCx \Gamma] a {\Modify [\nu] {A}} <n> \\
    \IsTy [\ECx \Gamma {\Modify [\nu] {A}}] B \\
    \IsTm [\ECx \Gamma A <\mu \circ \nu>] b {\Sb B {\ESb \Wk {\MkBox [\nu] {\Var {0}}}}}}
  {\IsTm {\LetMod a b <\nu> [\mu]} {\Sb B {\ESb \ISb a}}} \and
  \inferrule [term/sys-bin]
  {\IsCx \Gamma \\
    \IsTy A \\
    \IsTm {\phi , \psi} {\F _ m} \\
    \EqTm {\phi \lor \psi} \Ftop {\F _ m} \\
    \IsTm [\RCx \Gamma] a {\Sb A \WkF} \\
    \IsTm [\RCx \Gamma [\psi]] b {\Sb A {\WkF [\psi]}} \\
    \EqTm [\RCx \Gamma [\phi \land \psi]] {\Sb a {\RCx {\WkF [\phi \land \psi]}}}
    {\Sb b {\RCx {\WkF [\phi \land \psi]} [\psi]}} {\Sb A {\WkF [\phi \land \psi]}}}
  {\IsTm {\SysBin \phi a \psi b} A} \and
  \inferrule [term/sys-null]
  {\IsCx \Gamma \\
    \IsTy A \\
    \EqTm \Fbot \Ftop {\F _ m}}
  {\IsTm \SysNull A} \and
  \inferrule [term/comp]
  {\IsCx \Gamma \\
    \IsTy [\ICx \Gamma] A \\
    \IsTm \phi {\F_m} \\
    \IsTm [\ICx {\RCx \Gamma}] u {\Sb A {\ICx \WkF}} \\
    \IsTm {u_0} {\Sb A {\ESb {\ISb} {0}}} \\
    \EqTm [\RCx \Gamma] {\Sb u {\ESb {\ISb} {0}}}
    {\Sb {u_0} {\WkF}}
    {\Sb A {\ESb {\WkF} {0}}}}
  {\IsTm {\comp \, [\phi \mapsto u] \, u_0} {\Sb A {\ESb {\ISb} {1}}}} \and
  \inferrule [term/var]
  {\Mor [\mu] n m \\
    \IsCx \Gamma \\
    \IsTy [\LockCx \Gamma] A <n>}
  {\IsTm [\LockCx {\ECx \Gamma A}] {\Var {0}} {\Sb A {\LockSb \Wk}} <n>} \and
  \inferrule [term/sb]
  {\IsCx {\Gamma , \Delta} \\
    \IsSb \delta \Delta \\
    \IsTm [\Delta] a A}
  {\IsTm {\Sb a \delta} {\Sb A \delta}}
\end{mathparpagebreakable}

\paragraph*{Term equality.}\hfill

\begin{mathparpagebreakable}
  \inferrule [term-eq/pi-beta]
  {\Mor [\mu] n m \\
    \IsCx \Gamma \\
    \IsTy [\LockCx \Gamma] A <n> \\
    \IsTy [\ECx \Gamma A] B \\
    \IsTm [\LockCx \Gamma] a A <n> \\
    \IsTm [\ECx \Gamma A] b B}
  {\EqTm {\App {\Lam b} a} {\Sb b {\ESb \ISb a}} {\Sb B {\ESb \ISb a}}} \and
  \inferrule [term-eq/pi-eta]
  {\Mor [\mu] n m \\
    \IsCx \Gamma \\
    \IsTy [\LockCx \Gamma] A <n> \\
    \IsTy [\ECx \Gamma A] B \\
    \IsTm f {\Fn A B}}
  {\EqTm f {\Lam {\App {\Sb f \Wk} {\Var {0}}}} {\Fn A B}} \and
  \inferrule [term-eq/path-beta]
  {\IsCx \Gamma \\
    \IsTy [\ICx \Gamma] A \\
    \IsTm [\ICx \Gamma] a A \\
    \IsTm r {\I _ m}}
  {\EqTm {\PathApp {\PathAbs a} r} {\Sb a {\ESb \ISb r}} {\Sb A {\ESb \ISb r}}} \and
  \inferrule [term-eq/path-eta]
  {\IsCx \Gamma \\
    \IsTy [\ICx \Gamma] A \\
    \IsTm {a _ 0} {\Sb A {\ESb \ISb 0}} \\
    \IsTm {a _ 1} {\Sb A {\ESb \ISb 1}} \\
    \IsTm p {\Path A {a _ 0} {a _ 1}}}
  {\EqTm p {\PathAbs {\PathApp {\Sb p \WkI} \VarI}} {\Path A {a _ 0} {a _ 1}}} \and
  \inferrule [term-eq/mod-beta]
  {\Mor [\mu] n m \\
    \Mor [\nu] o n \\
    \IsCx \Gamma \\
    \IsTy [\LockCx {\LockCx \Gamma} <\nu>] A <o> \\
    \IsTy [\ECx \Gamma {\Modify [\nu] {A}}] B \\
    \IsTm [\LockCx {\LockCx \Gamma} <\nu>] a A <o> \\
    \IsTm [\ECx \Gamma A <\mu \circ \nu>] b {\Sb B {\ESb \Wk {\MkBox [\nu] {\Var {0}}}}}}
  {\EqTm {\LetMod {\MkBox [\nu] a} b <\nu> [\mu]} {\Sb b {\ESb \ISb a}} {\Sb B {\ESb \ISb {\MkBox
  [\nu] a}}}} \and
  \inferrule [term-eq/ext-type-beta]
  {\IsCx {\Gamma , \Delta} \\
    \IsSb \delta \Delta \\
    \IsTy [\LockCx \Delta] A <n> \\
    \IsTm [\LockCx \Gamma] a {\Sb A {\LockSb \delta}} <n>}
  {\EqTm [\LockCx \Gamma] {\Sb {\Var {0}} {\LockSb {\ESb \delta a}}} a {\Sb A {\LockSb \delta}}}
  \and
  \inferrule [term-eq/sys-top]
  {\IsCx \Gamma \\
    \IsTy A \\
    \IsTm \phi {\F _ m} \\
    \IsTm [\RCx \Gamma [\Ftop]] a {\Sb A {\WkF [\Ftop]}} \\
    \IsTm [\RCx \Gamma] b {\Sb A \WkF} \\
    \EqTm [\RCx \Gamma] {\Sb a {\RCx \WkF [\Ftop]}} b A}
  {\EqTm {\SysBin \Ftop a \phi b} {\Sb a {\RCx \ISb [\Ftop]}} A} \and
  \inferrule [term-eq/comp-face]
  {\IsCx \Gamma \\
    \IsTy [\ICx \Gamma] A \\
    \IsTm \phi {\F_m} \\
    \IsTm [\ICx {\RCx \Gamma}] u {\Sb A {\ICx \WkF}} \\
    \IsTm {u_0} {\Sb A {\ESb {\ISb} {0}}} \\
    \EqTm [\RCx \Gamma] {\Sb u {\ESb {\ISb} {0}}}
    {\Sb {u_0} {\WkF}}
    {\Sb A {\ESb {\WkF} {0}}} \\
    \EqTm {\phi} {\top} {\F_m}}
  {\EqTm {\comp \, [\phi \mapsto u] \, u_0}
    {\Sb u {\ESb {\RCx \ISb} {1}}}
    {\Sb A {\ESb {\ISb} {1}}}} \and
  \inferrule [term-eq/comp-mod${ }^\dagger$]
  {\IsCx \Gamma \\
    \Mor [\mu] n m \\
    \IsTy [\LockCx {\ICx \Gamma}] {A} <n> \\
    \IsTm \phi {\F_m} \\
    \IsTm [\LockCx {\ICx {\RCx \Gamma}}] u {\Sb A {\LockSb {\ICx \WkF}}} <n> \\
    \IsTm [\LockCx \Gamma] {u_0} {\Sb A {\LockCx {\ESb {\ISb} {0}}}} <n> \\
    \EqTm [\LockCx {\RCx \Gamma}] {\Sb u {\LockSb {\ESb {\ISb} {0}}}}
    {\Sb {u_0} {\LockSb \WkF}}
    {\Sb A {\LockSb {\ESb {\WkF} {0}}}} <n>}
  {\EqTm {\MkBox {\comp \, [\ExcCofInv \phi \mapsto \Sb u {\ExcSbI \circ \ExcSbR}] \,
        u_0}}
    {\comp \, [\phi \mapsto \MkBox{u}]\, \MkBox{u_0}}
    {\Sb {\Modify{A}} {\ESb {\ISb} {1}}}} \and
  \inferrule [term-eq/comp-pi${ }^\dagger$]
  {\IsCx \Gamma \\
    \Mor [\mu] n m \\
    \IsTy [\LockCx {\ICx \Gamma}] A <n> \\
    \IsTy [\ECx {\ICx \Gamma} A <\mu>] B <m> \\
    \IsTm \phi {\F_m} \\
    \IsTm [\ICx {\RCx \Gamma}] f {\Sb {(\Fn A B)} {\ICx \WkF}} \\
    \IsTm {f_0} {\Sb {(\Fn A B)} {\ESb {\ISb} {0}}} \\
    \EqTm [\RCx \Gamma] {\Sb f {\ESb {\ISb} {0}}}
      {\Sb {f_0} {\WkF}}
      {\Sb {(\Fn A B)} {\ESb {\WkF} {0}}} \\
    \IsTm [\LockCx \Gamma] {a_1} {\Sb A {\LockSb {\ESb {\ISb} 1}}} <n> \\
    \DefEqTm [\LockCx {\ICx \Gamma}] {w}
      {\Sb {\prn {\FILL \, [\, ] \, {a_1}}} {\ExcSbI}}
      {\Sb {A} {\LockSb{\Iinv[\VarI]}}} <n> \\
    \DefEqTm [\LockCx {\ICx \Gamma}] {v} {\Sb {w} {\LockSb {\Iinv [\VarI]}}} {A} <n>}
  {\EqTm {\App {(\comp \, [\phi \mapsto f] \, f_0)} {a_1}}
    {\comp \, [\phi \mapsto \App f {\Sb v {\LockSb {\ICx {\WkF}}}}]
      \, \App {f_0} {\Sb v {\LockSb {\ESb {\ISb} 0}}}}
    {\Sb B {\ESb {\ESb {\ISb} {1}} {a_1}}}} \and
  \inferrule [term-eq/face-res-bin]
  {\IsCx \Gamma \\
    \IsTm {\phi , \psi} {\F _ m} \\
    \EqTm {\phi \lor \psi} \Ftop {\F _ m} \\
    \IsTy A \\
    \IsTm {a , b} A \\
    \EqTm [\RCx \Gamma] {\Sb a \WkF} {\Sb b \WkF} {\Sb A \WkF} \\
    \EqTm [\RCx \Gamma [\psi]] {\Sb a {\WkF [\psi]}} {\Sb b {\WkF [\psi]}} {\Sb A {\WkF [\psi]}}}
  {\EqTm a b A} \and
  \inferrule [term-eq/face-res-null]
  {\IsCx \Gamma \\
    \EqTm \Fbot \Ftop {\F _ m} \\
    \IsTy A \\
    \IsTm {a , b} A}
  {\EqTm a b A}
\end{mathparpagebreakable}

\paragraph*{Derived Definitions.}\hfill

\begin{mathparpagebreakable}
  \inferH {sb/plus-int}
  {\IsCx {\Gamma , \Delta} \\
    \IsSb \delta \Delta}
  {\DefEqSb [\ICx \Gamma] {\ICx \delta} {\ESb {\prn {\delta \circ \WkI}} \VarI} {\ICx \Delta}} \and
  \inferrule [sb/exc-int${ }^\dagger$]
  {\Mor [\mu] n m \\
    \IsCx \Gamma}
  {\DefEqSb [\LockCx {\ICx \Gamma}] {\ExcSbI} {\ESb {\LockSb \WkI} {\ExcIntInv {\prn {\VarI}}}} {\ICx
  {\LockCx
  \Gamma} <n>}} \and
  \inferrule [sb/exc-face${ }^\dagger$]
  {\Mor [\mu] n m \\
    \IsCx \Gamma \\
    \IsTm \phi {\F _ m}}
  {\DefEqSb [\LockCx {\RCx \Gamma}] {\ExcSbR} {\RCx {\LockSb \WkF} [\ExcCofInv \phi]} {\RCx
  {\LockCx \Gamma} [\ExcCofInv \phi]}}
\end{mathparpagebreakable}

\section{Models of \texorpdfstring{\CubicalMTT{}}{Cubical MTT}}
\label{app:models}

\begin{defi}
  A modal context structure on a mode theory $\Mode$ is a strict 2-functor
  $\Mor[\Interp{-}]{\Coop{\Mode}}{\CAT}$ such that for each mode $m : \Mode$, $\Interp{m}$ has a
  terminal object.
\end{defi}

\begin{defi}
  A modal natural model on a modal context structure consists of
  \begin{itemize}
    \item for each mode $m : \Mode$, a presheaf $\SynTy{m} : \PSH{\Interp{m}}$,
    \item for each mode $m : \Mode$, a presheaf $\SynTm{m} : \PSH{\Interp{m}}$,
    \item for each mode $m : \Mode$, a natural transformation
    $\Mor[\SynEl{m}]{\SynTm{m}}{\SynTy{m}}$,
  \end{itemize}
  such that
  \begin{itemize}
    \item for any modes $m,n : \Mode$ and modality $\Mor[\mu]{n}{m}$, it holds that
    $\Mor[\Pre{\Interp{\mu}}{\SynEl{n}}]{\Pre{\Interp{\mu}}{\SynTm{n}}}{\Pre{\Interp{\mu}}{\SynTy{n}}}$
    is a representable natural transformation.
  \end{itemize}
\end{defi}

The type formers are the same as those in \cite[Section 5.2]{gratzer:journal:2021} and
\cite{awodey:2018} except for identity types which we do not have and path types which will come
later.

\begin{defi}
  A modal interval structure on a modal context structure consists of
  \begin{itemize}
    \item for each mode $m : \Mode$, a De Morgan algebra $\IntOb{m} : \Interp{m}$,
    \item for any modes $m,n : \Mode$ and each modality $\Mor[\mu]{n}{m}$, a natural transformation
    of De Morgan algebras $\Mor[\IntArr{\mu}]{\Yo{\IntOb{m}}}{\Pre{\Interp{\mu}}{\Yo{\IntOb{n}}}}$,
  \end{itemize}
  such that
  \begin{itemize}
    \item the preshaves $\Yo{\IntOb{m}}$ and morphisms $\IntArr{\mu}$ assemble into a lax natural
    transformation $\Mor[\Mor{\Interp{-}}{\Op{\SET}}]{\Coop{\Mode}}{\CAT}$, where
    $\Mor[\Op{\SET}]{\Coop{\Mode}}{\CAT}$ is the functor constantly equal to $\Op{\SET}$,
    \item for each mode $m : \Mode$ and context $\Gamma : \Interp{m}$, the product
    $\Gamma\times\IntOb{m}$ exists,
    \item for any modes $m,n : \Mode$, each modality $\Mor[\mu]{n}{m}$, and each context $\Gamma :
    \Interp{m}$, the uniquely determined dashed arrow in the following diagram has an inverse:
    \[\begin{tikzpicture}[diagram]
      \node (LN) {$\Gamma$};
      \node [below = of LN] (LS) {$\Gamma\times\IntOb{m}$};
      \node [right = 4cm of LN] (NW) {$\Interp{\mu}\Gamma$};
      \node [below = of NW] (SW) {$\Interp{\mu}\prn{\Gamma\times\IntOb{m}}$};
      \node [right = 3cm of NW] (NE) {$\Interp{\mu}\Gamma\times\IntOb{n}$};
      \node [right = 3cm of SW] (SE) {$\IntOb{n}$};
      \node [below = 2.5cm of SE] (BE) {$\IntOb{m}$};
      \node [left = 3cm of BE] (BW) {$\Gamma\times\IntOb{m}$};
      \path [->] (LS) edge node [left] {$\pi_1$} (LN);
      \path [->] (SW) edge node [left] {$\Interp{\mu}\pi_1$} (NW);
      \path [->] (BW) edge node [below] {$\pi_2$} (BE);
      \path [->] (SW) edge node [below] {$\IntArr{\mu,\Gamma\times\IntOb{m}}\prn{\pi_2}$} (SE);
      \path [->] (NE) edge node [above] {$\pi_1$} (NW);
      \path [->] (NE) edge node [right] {$\pi_2$} (SE);
      \path [dashed,->] (SW) edge (NE);
      \node [below = 1cm of LN] (LM) {};
      \node [right = 0.3cm of LM] (LW) {};
      \node [right = 2.5cm of LW] (LE) {};
      \path [|->] (LW) edge node [above] {$\Interp{\mu}$} (LE);
      \node [left = 1.5cm of BE] (BM) {};
      \node [above = 0.3cm of BM] (BS) {};
      \node [above = 1.6cm of BS] (BN) {};
      \path [|->] (BS) edge node [right] {$\IntArr{\mu,\Gamma\times\IntOb{m}}$} (BN);
    \end{tikzpicture}\]
  \end{itemize}
\end{defi}

\begin{defi}
  A modal face structure on a modal interval structure consists of
  \begin{itemize}
    \item a lax natural transformation
    $\Mor[\Mor[\FPSH]{\Interp{-}}{\Op{\SET}}]{\Coop{\Mode}}{\CAT}$, where
    $\Op{\SET}$ is the functor constantly equal to $\Op{\SET}$,
    \item for each mode $m : \Mode$, a natural transformation
    $\Mor[\EqZ{m}]{\Yo{\IntOb{m}}}{\Op{\FPSH{m}}}$,
  \end{itemize}
  such that
  \begin{itemize}
    \item $\FPSH$ factors through $\Op{\BDISLAT}$, the functor constantly equal to the opposite of
    the category of bounded distributive lattices,
    \item for each mode $m : \Mode$, each context $\Gamma : \Interp{m}$, and each interval term
    $\Mor[r]{\Gamma}{\IntOb{m}}$, it holds that $\FPSH*{\mu}{\Gamma}{\EqZ{m}{\Gamma}{r}} =
    \EqZ{n}{\Interp{\mu}\prn{\Gamma}}{\IntArr{\mu,\Gamma}\prn{r}}$,
    \item for each mode $m : \Mode$ and each context $\Gamma : \Interp{m}$, it holds that
    $\EqZ{m}{\Gamma}{0} = \Ftop$, where $0$ is from $\IntOb{m}$ being a De Morgan algebra, and
    $\Ftop$ is from $\FPSH{m}{\Gamma}$ being a bounded lattice,
    \item for each mode $m : \Mode$, each context $\Gamma : \Interp{m}$, and each interval term
    $\Mor[r]{\Gamma}{\IntOb{m}}$, it holds that $\EqZ{m}{\Gamma}{r}\land\EqZ{m}{\Gamma}{\Iinv[r]} =
    \Fbot$, where $\Iinv[r]$ is from $\IntOb{m}$ being a De Morgan algebra, and $\land$ and $\Fbot$
    are from $\FPSH{m}{\Gamma}$ being a bounded lattice.
  \end{itemize}
\end{defi}

\begin{defi}
  A modal restriction structure on a modal face structure consists of
  \begin{itemize}
    \item for each mode $m : \Mode$, each context $\Gamma : \Interp{m}$, and each face
    $\Mor[\phi]{\Yo{\Gamma}}{\FPSH{m}}$, a choice of pullback of the form:
    \[\begin{tikzpicture}[diagram]
      \node (SE) {$\FPSH{m}$};
      \node [left = of SE] (SW) {$\Yo{\Gamma}$};
      \node [above = 1.5cm of SE] (NE) {$1$};
      \node [above = 1.5cm of SW] (NW) {$\Yo{\Restr{m}{\Gamma}{\phi}}$};
      \path [->] (NE) edge node [right] {$\Ftop$} (SE);
      \path [->] (SW) edge node [below] {$\phi$} (SE);
      \path [->] (NW) edge (NE);
      \path [->] (NW) edge node [left] {$\Yo{\WeakF{m}{\Gamma}{\phi}}$} (SW);
    \end{tikzpicture}\]
  \end{itemize}
  such that
  \begin{itemize}
    \item for all modes $m,n : \Mode$, each modality $\Mor[\mu]{n}{m}$, each context $\Gamma :
    \Interp{m}$, and each face $\Mor[\phi]{\Yo{\Gamma}}{\FPSH{m}}$, the uniquely determined dashed
    arrow in the following diagram has an inverse:
    \[\begin{tikzpicture}[diagram]
      \node (SE) {$\FPSH{n}$};
      \node [left = 3cm of SE] (SW) {$\Yo{\Interp{\mu}\Gamma}$};
      \node [above = of SE] (NE) {$1$};
      \node [above = of SW] (NW)
      {$\Yo{\Restr{n}{\Interp{\mu}\Gamma}{\FPSH*{\mu}{\Gamma}{\phi}}}$};
      \path [->] (NE) edge node [right] {$\Ftop$} (SE);
      \path [->] (SW) edge node [below] {$\FPSH*{\mu}{\Gamma}{\phi}$} (SE);
      \path [->] (NW) edge (NE);
      \path [->] (NW) edge node [left]
      {$\Yo{\WeakF{n}{\Interp{\mu}\Gamma}{\FPSH*{\mu}{\Gamma}{\phi}}}$} (SW);
      \node [above = 1cm of NW] (NNW) {};
      \node [left = 2.8cm of NNW] (NNWW) {$\Yo{\Interp{\mu}\prn{\Restr{m}{\Gamma}{\phi}}}$};
      \path [->,bend left=11] (NNWW) edge (NE);
      \path [->,bend right=39] (NNWW) edge node [left] {$\Yo{\Interp{\mu}\WeakF{m}{\Gamma}{\phi}}$}
      (SW);
      \path [->,dashed] (NNWW) edge (NW);
    \end{tikzpicture}\]
    Here, the commutativity of the outer square follows from the following calculation:
    \begin{align*}
      \FPSH*{\mu}{\Gamma}{\phi}\circ\Yo{\Interp{\mu}\WeakF{m}{\Gamma}{\phi}} & =
      \FPSH{n}{\Interp{\mu}\WeakF{m}{\Gamma}{\phi}}{\FPSH*{\mu}{\Gamma}{\phi}} \\
      & =
      \prn{\prn{\FPSH{n}\circ\Interp{\mu}}\prn{\WeakF{m}{\Gamma}{\phi}}\circ\FPSH*{\mu}{\Gamma}}\prn{\phi}
       \\
      & =
      \prn{\FPSH*{\mu}{\Restr{m}{\Gamma}{\phi}}\circ\FPSH{m}{\WeakF{m}{\Gamma}{\phi}}}\prn{\phi} \\
      & = \FPSH*{\mu}{\Restr{m}{\Gamma}{\phi}}{\Yo{\WeakF{m}{\Gamma}{\phi}}\circ\phi} \\
      & = \FPSH*{\mu}{\Restr{m}{\Gamma}{\phi}}{\Ftop} \\
      & = \Ftop.
    \end{align*}
  \end{itemize}
\end{defi}

\begin{rem}
  For each mode $m : \Mode$, each context $\Gamma : \Interp{m}$, and any faces
  $\Mor[\phi,\psi]{\Yo{\Gamma}}{\FPSH{m}}$ with $\phi\leq\psi$, consider the following diagram:
  \[\begin{tikzpicture}[diagram]
    \node (SE) {$\FPSH{m}$};
    \node [left = of SE] (SW) {$\Yo{\Gamma}$};
    \node [above = 1.5cm of SE] (NE) {$1$};
    \node [above = 1.5cm of SW] (NW) {$\Yo{\Restr{m}{\Gamma}{\psi}}$};
    \path [->] (NE) edge node [right] {$\Ftop$} (SE);
    \path [->,bend left=20] (SW) edge node [below] {$\psi$} (SE);
    \path [->] (NW) edge (NE);
    \path [->] (NW) edge node [left] {$\Yo{\WeakF{m}{\Gamma}{\psi}}$} (SW);
    \node [above = 1cm of NW] (NNW) {};
    \node [left = 2.3cm of NNW] (NNWW) {$\Yo{\Restr{m}{\Gamma}{\phi}}$};
    \path [->,bend left=14] (NNWW) edge (NE);
    \path [->,bend right=37] (NNWW) edge node [left] {$\Yo{\WeakF{m}{\Gamma}{\phi}}$} (SW);
    \path [->,bend right=20] (SW) edge node [below] {$\phi$} (SE);
  \end{tikzpicture}\]
  We can calculate
  \begin{align*}
    \psi\circ\Yo{\WeakF{m}{\Gamma}{\phi}} & = \FPSH{m}{\WeakF{m}{\Gamma}{\phi}}{\psi} \\
    & \geq \FPSH{m}{\WeakF{m}{\Gamma}{\phi}}{\phi} \\
    & = \phi\circ\Yo{\WeakF{m}{\Gamma}{\phi}} \\
    & = \Ftop,
  \end{align*}
  and thus $\psi\circ\Yo{\WeakF{m}{\Gamma}{\phi}} = \Ftop$, implying the outer square commutates,
  and we thus get a canonical morphism $\Mor{\Restr{m}{\Gamma}{\phi}}{\Restr{m}{\Gamma}{\psi}}$.
\end{rem}

\begin{defi}
  A modal face sturcture and a modal natural model (both on the same modal context sturcture) has
  systems if
  \begin{itemize}
    \item for each mode $m : \Mode$, each context $\Gamma : \Interp{m}$, any faces $\phi,\psi :
    \FPSH{m}{\Gamma}$, each presheaf $X : \PSH{\Interp{m}}$, and each commuting diagram
    \begin{equation*}
      \begin{tikzpicture}[diagram]
        \node (NW) {$\Yo{\Restr{m}{\Gamma}{\phi\land\psi}}$};
        \node [right = 4cm of NW] (N) {$\Yo{\Restr{m}{\Gamma}{\psi}}$};
        \node [below = of NW] (W) {$\Yo{\Restr{m}{\Gamma}{\phi}}$};
        \node [right = 4cm of W] (C) {$\Yo{\Restr{m}{\Gamma}{\phi\lor\psi}}$};
        \node [right = 2cm of C] (E) {};
        \node [below = 1cm of E] (SE) {$X$};
        \path [->] (NW) edge (N);
        \path [->] (NW) edge (W);
        \path [->] (N) edge (C);
        \path [->] (W) edge (C);
        \path [->] (W) edge [bend right = 9] (SE);
        \path [->] (N) edge [bend left] (SE);
      \end{tikzpicture}
    \end{equation*}
    where the arrows in the inner square are the canonical morphisms following from
    $\phi\land\psi\leq\phi$, $\phi\land\psi\leq\psi$, $\phi\leq\phi\lor\psi$, and
    $\psi\leq\phi\lor\psi$, if $X$ is representable or $\FPSH{m}$ there exists at most one morphism
    $\Mor{\Yo{\Restr{m}{\Gamma}{\phi\lor\psi}}}{X}$ such that the diagram commutes, and if $X$ is
    $\SynTy{m}$ or $\SynTm{m}$ there exists exactly one such morphism,
    \item for each mode $m : \Mode$, each context $\Gamma : \Interp{m}$, and each presheaf $X :
    \PSH{\Interp{m}}$, if $X$ is representable, $\FPSH{m}$, or $\SynTy{m}$ there exists at most one
    morphism $\Mor{\Yo{\Restr{m}{\Gamma}{\Fbot}}}{X}$, and if $X$ is $\SynTm{m}$ there exists
    exactly one such morphism.
  \end{itemize}
\end{defi}

\begin{defi}
  A path structure on modal interval structure and a modal natural model (both on the same modal
  context structure) is a direct translation of the rules for path types, and we will thus not give
  the details.
\end{defi}

\begin{defi}
  A composition structure on modal restriction structure is a direct translation of the rules for
  composition, and we will thus not give the details.
\end{defi}

\end{document}